\documentclass[12pt,english,american]{article}
\usepackage[T1]{fontenc}
\usepackage[latin9]{inputenc}
\usepackage{geometry}
\usepackage{color}
\usepackage{babel}
\usepackage{verbatim}
\usepackage{refstyle}
\usepackage{textcomp}
\usepackage{amsthm}
\usepackage{amsmath}
\usepackage{amssymb}
\usepackage[unicode=true,pdfusetitle,
 bookmarks=true,bookmarksnumbered=false,bookmarksopen=false,
 breaklinks=true,pdfborder={0 0 0},backref=false,colorlinks=true]
 {hyperref}
\hypersetup{
 linkcolor=blue,citecolor=blue, urlcolor=blue}

\makeatletter

\AtBeginDocument{\providecommand\eqref[1]{\ref{eq:#1}}}
\AtBeginDocument{\providecommand\lemref[1]{\ref{lem:#1}}}
\AtBeginDocument{\providecommand\thmref[1]{\ref{thm:#1}}}
\AtBeginDocument{\providecommand\defref[1]{\ref{def:#1}}}
\AtBeginDocument{\providecommand\secref[1]{\ref{sec:#1}}}
\newcommand{\noun}[1]{\textsc{#1}}
\RS@ifundefined{subref}
  {\def\RSsubtxt{section~}\newref{sub}{name = \RSsubtxt}}
  {}
\RS@ifundefined{thmref}
  {\def\RSthmtxt{theorem~}\newref{thm}{name = \RSthmtxt}}
  {}
\RS@ifundefined{lemref}
  {\def\RSlemtxt{lemma~}\newref{lem}{name = \RSlemtxt}}
  {}

\usepackage{enumitem}
\newlength{\lyxlabelwidth}
\usepackage[color=green!40]{todonotes}
\theoremstyle{plain}
\newtheorem{thm}{\protect\theoremname}[section]
\theoremstyle{definition}
\newtheorem{defn}[thm]{\protect\definitionname}
\theoremstyle{remark}
\newtheorem{rem}[thm]{\protect\remarkname}

\newenvironment{elabeling}[2][]%
{\settowidth{\lyxlabelwidth}{#2}
\begin{description}[font=\normalfont,style=sameline,
leftmargin=\lyxlabelwidth,#1]}
{\end{description}}
\theoremstyle{plain}
\newtheorem{prop}[thm]{\protect\propositionname}
\theoremstyle{plain}
\newtheorem{lem}[thm]{\protect\lemmaname}
\ifx\proof\undefined
\newenvironment{proof}[1][\protect\proofname]{\par
\normalfont\topsep6\p@\@plus6\p@\relax
\trivlist
\itemindent\parindent
\item[\hskip\labelsep
\scshape
#1]\ignorespaces
}{%
\endtrivlist\@endpefalse
}
\providecommand{\proofname}{Proof}
\fi
\theoremstyle{plain}

\usepackage{txfonts,bbm}

\newcommand{\charf}{\mathbbm 1}

\newref{con}{name = Conjecture\ }
\newref{prop}{name = Proposition\ }
\newref{def}{name = Definition\ }
\newref{sec}{name = Section\ }
\newref{thm}{name = Theorem\ }
\newref{lem}{name = Lemma\ }
\newref{cor}{name = Corollary\ }

\makeatother

\addto\captionsamerican{\renewcommand{\corollaryname}{Corollary}}
\addto\captionsamerican{\renewcommand{\definitionname}{Definition}}
\addto\captionsamerican{\renewcommand{\lemmaname}{Lemma}}
\addto\captionsamerican{\renewcommand{\propositionname}{Proposition}}
\addto\captionsamerican{\renewcommand{\remarkname}{Remark}}
\addto\captionsamerican{\renewcommand{\theoremname}{Theorem}}
\addto\captionsenglish{\renewcommand{\corollaryname}{Corollary}}
\addto\captionsenglish{\renewcommand{\definitionname}{Definition}}
\addto\captionsenglish{\renewcommand{\lemmaname}{Lemma}}
\addto\captionsenglish{\renewcommand{\propositionname}{Proposition}}
\addto\captionsenglish{\renewcommand{\remarkname}{Remark}}
\addto\captionsenglish{\renewcommand{\theoremname}{Theorem}}
\providecommand{\corollaryname}{Corollary}
\providecommand{\definitionname}{Definition}
\providecommand{\lemmaname}{Lemma}
\providecommand{\propositionname}{Proposition}
\providecommand{\remarkname}{Remark}
\providecommand{\theoremname}{Theorem}

\begin{document}

\title{\noun{Ultraviolet Properties of }\\
\noun{ the Spinless, One-Particle Yukawa Model}}

\author{D.-A. Deckert, A. Pizzo}

\maketitle
\global\long\def\gap#1#2{\textrm{Gap}\left(#1\upharpoonright#2\right)}

\global\long\def\pmax{P_{\mathrm{max}}}

\global\long\def\restrict{\upharpoonright}

\global\long\def\spec#1{\mathrm{spec}\left(#1\right)}

\begin{abstract}
We consider the one-particle sector of the spinless Yukawa model,
which describes the interaction of a nucleon with a real field of
scalar massive bosons (neutral mesons). The nucleon as well as the
mesons have relativistic dispersion relations. In this model we study
the dependence of the nucleon mass shell on the ultraviolet cut-off
$\Lambda$. For any finite ultraviolet cut-off the nucleon one-particle states
are constructed in a bounded region of the energy-momentum space. We identify the dependence of the ground state energy
on $\Lambda$ and the coupling constant. More importantly, we show
that the model considered here becomes essentially trivial in the
limit $\Lambda\to\infty$ regardless of any (nucleon) mass and self-energy
renormalization. Our results hold in the small coupling regime.\\
\\
\noindent{\bf Acknowledgments:} D.-A.D. gratefully acknowledges financial support from the post-doc program of the DAAD. A.P. is supported by the NSF grant \#DMS-0905988. \end{abstract}
\tableofcontents{}

\section{Introduction and Definition of the Model}

The Yukawa theory provides an effective description of the strong
nuclear forces between massive nucleons which are mediated by mesons.
The nucleons as well as the mesons have relativistic dispersion relations.
It is well-known that the Yukawa theory is plagued by ultraviolet
divergences, and so far the fully relativistic model has only been
constructed in $1+1$ dimensions; see \cite{Mastropietro_2008_Non-perturbative-renormalization}
and references therein for the details.

In this paper we consider a toy model of the Yukawa theory, referred
to as \emph{spinless, one-particle Yukawa model}, obtained by neglecting
pair-creation and spin, and we restrict the analysis to the one-nucleon
sector. In order to yield a well-defined Hamiltonian for this model
one usually introduces a cut-off which removes the problematic meson
momenta from the interaction term above a finite threshold energy
$\Lambda$. While for non-relativistic situations one may argue that
a cut-off $\Lambda$ of the order of the nucleon rest mass should
render a satisfying predictive power of the model, a finite cut-off
is not justified in the relativistic regime. Though the model we deal
with is a caricature of the relativistic interaction between nucleons
and mesons, we address the mathematical problem how to control the
model uniformly in $\Lambda$ beyond perturbation theory. 

More specifically, we analyze the effect of self-energy and mass renormalization
in the limit $\Lambda\to\infty$. It is a common hope that at least
for non-relativistic QED, i.e., for the Pauli-Fierz Hamiltonian, the
ultraviolet cut-off can possibly be removed by introducing a suitable
mass and energy renormalization; see \cite{spohn_dynamics_2004}.
The believe is that, in contrast to classical electrodynamics where
the electron bare mass is sent to negative infinity, in non-relativistic
QED the bare mass should tend to zero as $\Lambda\to\infty$ to compensate
for the growing electrodynamic mass. Our results show that because
of the relativistic dispersion relation of the nucleon this is not
the case for the spinless, one-particle Yukawa model. Namely, in a
neighborhood of the origin of the (total) momentum space and for small
values of the coupling constant, we establish two goals:
\begin{enumerate}
\item We identify the dependence of the ground state energy on $\Lambda$
and the coupling constant $g$.
\item We show that the nucleon mass shell becomes flat in the limit $\Lambda\to\infty$
up to corrections estimated to be $\mathcal{O}_{g\to0}(|g|^{\frac{1}{2}})$,
irrespectively of any scaling of the (nucleon) bare mass $m$, i.e., $m\equiv m(\Lambda)>0$.

\end{enumerate}
Our analysis is based on a multi-scale technique which was developed
in \cite{pizzo_one-particle_2003} to treat the infrared divergence
of the Nelson model, and which was recently refined in \cite{bachmann_mass_2011}
to simultaneously control the infrared and ultraviolet divergences
of the same model. We extend this multi-scale technique further and
apply it to the spinless, one-particle Yukawa model. 

It is interesting to note
that for this model the self-energy diverges
linearly for $\Lambda\to\infty$ as it is the case for its classical
analogue.

\paragraph{Definition of the Model. }

The Hilbert space of the model is
\[
\mathcal{H}:=L^{2}(\mathbb{R}^{3},\mathbb{C};dx)\otimes\mathcal{F}(h)\,,
\]
where $\mathcal{F}(h)$ is the Fock space of scalar bosons
\[
\mathcal{F}(h):=\bigoplus_{j=0}^{\infty}\mathcal{F}^{(j)}\,,\quad\quad\quad\quad\mathcal{F}^{(0)}:=\mathbb{C}\,,\quad\quad\quad\quad\mathcal{F}^{j\geq1}:=\bigodot_{l=1}^{j}h\,,\quad\quad\quad\quad h:=L^{2}(\mathbb{R}^{3},\mathbb{C};dk)
\]
where $\odot$ denotes the symmetric tensor product. Let $a(k),a^{*}(k)$
be the usual Fock space annihilation and creation operators satisfying
the canonical commutation relations (CCR)
\[
\left[a(k),a(q)^{*}\right]=\delta(k-q),\qquad\left[a(k),a(q)\right]=0=\left[a(k)^{*},a^{*}(q)\right],\qquad\forall k,q\in\mathbb{R}^{3}.
\]
The kinematics of the system is described by: (a) The position $x$
and the momentum $p$ of the nucleon that satisfy the Heisenberg commutation
relations. (b) The real scalar field $\Phi$ and its conjugate momentum.

The dynamics is generated by the Hamiltonian 
\begin{equation}
H|_{\kappa}^{\Lambda}:=\sqrt{p^{2}+m^{2}}+H^{f}+g\Phi|_{\kappa}^{\Lambda}(x)\label{eq:hamiltonian}
\end{equation}
where:
\begin{itemize}
  \item $m$ is the nucleon mass;
  \item $g\in\mathbb{R}$ is the coupling constant;
  \item 
  \[
H^{f}:=\int dk\,\omega(k)a^{*}(k)a(k),\qquad\omega(k)\equiv\omega(|k|):=\sqrt{|k|^{2}+\mu^{2}},
\]
is the free field Hamiltonian with $\mu$ being the meson mass;
  \item 
the interaction term is given by
\begin{equation}
\Phi|_{\kappa}^{\Lambda}(x):=\phi|_{\kappa}^{\Lambda}(x)+\phi^{*}|_{\kappa}^{\Lambda}(x),\qquad\phi|_{\kappa}^{\Lambda}(x):=\int_{\mathcal{B}_{\Lambda}\setminus\mathcal{B}_{\kappa}}dk\,\rho(k)a(k)e^{ikx},\qquad\rho(k):=\frac{1}{(2\pi)^{3/2}}\frac{1}{\sqrt{2\omega(k)}}\label{eq:form_factor}
\end{equation}
for $0\leq\kappa<\Lambda$, and for the domain of integration we use the notation $\mathcal{B}_{\sigma}:=\{k\in\mathbb{R}^{3}\,|\:|k|<\sigma\}$
for any $\sigma>0$;
\item we use units such that $\hbar=c=1$.
\end{itemize}

Note that for $\Lambda=\infty$ the formal expression of the interaction
$\Phi|_{\kappa}^{\Lambda}(x)$ is not a well-defined operator on $\mathcal{H}$
because the form factor $\rho(k)$ is not square integrable. It is
well-known (see also Proposition 1.1 below) that for $0\leq\kappa<\Lambda<\infty$
the operator $H|_{\kappa}^{\Lambda}$ is self-adjoint and its domain
coincides with the one of $H^{(0)}:=\sqrt{p^{2}+m^{2}}+H^{f}$

We briefly recall some well-known facts about this model. The total
momentum operator of the system is 
\begin{equation}
P:=p+P^{f}:=p+\int dk\, a^{\text{*}}(k)a(k)\label{eq:total_momentum}
\end{equation}
where $P^{f}$ is the field momentum. Due to translational invariance
of the system the Hamiltonian and the total momentum operator commute.
Hence, the Hilbert space $\mathcal{H}$ can be decomposed on the joint
spectrum of the three components of the total momentum operator, i.e.,
\[
\mathcal{H=\int}^{\oplus}dP\,\mathcal{H}_{P}
\]
here $\mathcal{H}_{P}$ is a copy of the Fock space $\mathcal{F}$
carrying the (Fock) representation corresponding to annihilation and
creation operators 
\[
b(k):=a(k)e^{ikx},\quad b^{*}(k):=a^{*}(k)e^{-ikx}\,.
\]
We will use the same symbol $\mathcal{F}$ for all Fock spaces. The
fiber Hamiltonian can be expressed as 
\[
H_{P}|_{\kappa}^{\Lambda}:=\sqrt{(P-P^{f})^{2}+m^{2}}+H^{f}+g\Phi|_{\kappa}^{\Lambda}
\]
where
\[
\Phi|_{\kappa}^{\Lambda}:=\phi|_{\kappa}^{\Lambda}+\phi^{*}|_{\kappa}^{\Lambda},\qquad\phi|_{\kappa}^{\Lambda}:=\int_{\mathcal{B}_{\Lambda}\setminus\mathcal{B}_{\kappa}}dk\,\rho(k)b(k),
\]
and 
\[
H^{f}=\int dk\,\omega(k)b^{*}(k)b(k),\qquad P^{f}=\int dk\, kb^{*}(k)b(k).
\]
By construction, the fiber Hamiltonian maps its domain in $\mathcal{H}_{P}$
into $\mathcal{H}_{P}$. Finally, for later use we define 
\[
H_{P}^{(0)}:=H_{P}^{nuc}+H^{f},\quad\quad\quad H_{P}^{nuc}:=\sqrt{(P-P^{f})^{2}+m^{2}}.
\]
We restrict our study to the \textbf{model parameters}: 
\[
m>0,\qquad\mu>1,\qquad0<|g|\leq1,\qquad0<\kappa\leq1<\Lambda<\infty,\qquad0<P_{max}<\frac{1}{2},\qquad|P|<P_{max}.
\]
The choice $\mu>1$ and $P_{max}$ less than one is only a technical artifact
of the crude  estimate (\ref{eq:energy_difference}) in the proof
of \lemref{energy} which provides an easy spectral gap estimate in
\lemref{gap_estimate} that we employ in the multi-scale analysis.
\\

Concerning previous results on the spinless, one-particle Yukawa model
we refer the reader to \cite{eckmann_model_1970,fraehlich_infrared_1973,fraehlich_existence_1974,Takaesu_2009_Ground-States-of-the-Yukawa-models-with-Cutoffs}.
In \cite{eckmann_model_1970} Eckmann considers the spinless Yukawa
model without pair-creation with a regularization of the meson form
factor. In contrast to our choice given in (\ref{eq:form_factor})
the interaction term in his Hamiltonian is given by
\[
\underset{|p|,|k|,|p-k|\leq\Lambda}{\int dp\int dk}\frac{n^{*}(p-k)\, a^{*}(k)\, n(p)}{\sqrt{((p-k)^{2}+\mu^{2})^{1/2}(k^{2}+\mu^{2})^{1/2}(p{}^{2}+\mu^{2})^{1/2}}}+h.c.
\]
where $n^{*}(p)$ and $n(p)$ denote the nucleon creation and annihilation
operators. This implies that the Hamiltonian renormalized by means
of a mass operator (for details see \cite{eckmann_model_1970}) converges
in the norm resolvent sense as $\Lambda\to\infty$. Furthermore, in
\cite{eckmann_model_1970} the one-particle scattering states are
constructed in the small coupling regime. Also Fröhlich \cite{fraehlich_existence_1974}
studied the spinless, one-particle Yukawa model but with the meson
form factor $\frac{\rho(k)}{|k|^{1/2}}$, for which he showed that
the Hamiltonian including a logarithmically divergent self-energy
renormalization constant is well defined in the limit $\Lambda\to\infty$
and that the nucleon mass shell is non-trivial.

The behavior of the ground state energy for $\Lambda\to\infty$ has been addressed in
\cite{lieb_loss_self_energy} and \cite{hainzl_2002} for non-relativistic and pesudo-relativistic QED models.
In particular, in \cite{lieb_loss_self_energy},  for the relativistic dispersion relation the
electron self-energy has been proven to obey the same type of dependence on
$\Lambda$ as in our model, but without the restriction to the small coupling
regime. Perturbative mass renormalization in non-relativistic QED has been
addressed in \cite{hiroshima_2005}. Furthermore, mass renormalization based on the binding energy of hydrogen has been discussed in models of quantum electrodynamics in \cite{lieb_loss_hydrogen}.

We also want to mention \cite{Konenberg_2012_The-mass-shell-in-the-semi-relativistic-Pauli-Fierz-model} for a recent application 
of the iterative analytic perturbation theory to the so-called semi-relativistic
Pauli-Fierz model that focusses on the infrared corrections to the
electron mass shell.

\paragraph{Notation.}
\begin{enumerate}
\item The symbol $C$ denotes any positive universal constant and may change
its value from line to line.
\item The components of a vector $v\in\mathbb{R}^{3}$ are denoted by $v=(v_{1},v_{2},v_{3})$.
\item The bars $\left|\cdot\right|,\left\Vert \cdot\right\Vert $ denote
the euclidean and the Fock space norm, respectively.
\item The brackets \textlangle{}·, ·\textrangle{} denote the scalar product
of vectors in ${\cal F}$. Given a subspace ${\cal K}\subseteq{\cal F}$
and an operator $A$ on ${\cal F}$ we use the notation
\[
\left\Vert A\right\Vert _{{\cal K}}\equiv\left\Vert A\restrict{\cal K}\right\Vert _{{\cal F}}.
\]

\item A \emph{hat} over a vector means that the vector is of unit length,
i.e., $\widehat{\Psi}:=\frac{\Psi}{\left\Vert \Psi\right\Vert }.$
\item For two vectors $\psi,\chi$ we write $\psi\parallel\chi$ if they
are parallel and $\psi\perp\chi$ if they are perpendicular.
\item We denote the spectral gap of a self-adjoint operator $H$ restricted
to a subspace ${\cal K}\subseteq{\cal F}$ with unique ground state
$\Psi$ and corresponding ground state energy $E$ by
\[
\gap H{{\cal K}}:=\inf\spec{H\restrict{\cal K}}\setminus\left\{ E\right\} -E=\inf_{\psi\perp\Psi}\left\langle \widehat{\psi},(H-E)\widehat{\psi}\right\rangle 
\]
where the infimum is taken over the domain of $H\restrict{\cal K}$.
\item We use the short-hand notation ($\gamma$ is defined in (\ref{eq:gamma}))
\[
H_{P,n}:=H_{P}|_{\Lambda\gamma^{n},\qquad}^{\Lambda}\ldots|_{n}^{m}=\ldots|_{\Lambda\gamma^{n}}^{\Lambda\gamma^{m}},\qquad\int_{a}^{b}dk=\int_{{\cal B}_{b}\setminus{\cal B}_{a}}dk.
\]

\end{enumerate}

\section{Strategy and Main Results}

Our computations are based on von Neumann expansion formulas of the
ground state of the Hamiltonians $H_{P}|_{\kappa}^{\Lambda}$ by \emph{iterative
analytic perturbation theory}, that means by a multi-scale procedure
that relies on analytic perturbation theory. Indeed, in order to study
the $\Lambda$-dependence of the mass shell, we need to construct
the ground states for a fixed and non-zero value of $g$ that is independent
of the cut-off $\Lambda.$ Note however that unless the coupling constant
$g$ is of order $\left(\frac{1}{\Lambda}\right)^{\frac{1}{2}}$ one
cannot add the full interaction $g\Phi|_{\kappa}^{\Lambda}$ to the
free Hamiltonian $H_{P}^{(0)}$ in a single shot of perturbation theory.
Therefore, instead of adding the interaction in one shot we shall
do many intermediate steps in the expansion by slicing up the interaction
term of the Hamiltonian into smaller pieces, namely slices corresponding
to momentum ranges $[\Lambda\gamma^{n-1},\Lambda\gamma^{n})$ that
can be made arbitrarily thin by adjusting a fineness parameter $\gamma$
\begin{equation}
\frac{1}{2}<\gamma<1.\label{eq:gamma}
\end{equation}
It turns out that in this way one can maintain control over the convergence
radius of the von Neumann expansions uniformly in $\Lambda$. With
respect to this slicing we define the Fock spaces:
\begin{defn}
For $n\in\{0\}\cup\mathbb{N}$, we define the Fock spaces
\begin{eqnarray*}
{\cal F} & := & {\cal F}\left(L^{2}\left(\mathbb{R}^{3},\mathbb{C};dk\right)\right),\\
{\cal F}_{n} & := & {\cal F}\left(L^{2}\left(\mathbb{R}^{3}\setminus{\cal B}_{\Lambda\gamma^{n}},\mathbb{C};dk\right)\right),\qquad\\
{\cal F}|_{n}^{n-1} & := & {\cal F}\left(L^{2}\left({\cal B}_{\Lambda\gamma^{n-1}}\setminus{\cal B}_{\Lambda\gamma^{n}},\mathbb{C};dk\right)\right).
\end{eqnarray*}
In all these Fock spaces we shall use the same symbol $\Omega$ to
denote the vacuum. For a vector $\psi$ in $\mathcal{F}_{n-1}$ and
an operator $O$ on $\mathcal{F}_{n-1}$ we shall use the same symbol
to denote the vector $\psi\otimes\Omega$ in $\mathcal{F}_{n}$ and
the operator $O\otimes\charf_{\mathcal{F}_{n}^{n-1}}$ on $\mathcal{F}_{n}$,
respectively, where $\charf_{\mathcal{F}_{n}^{n-1}}$ is the identity
operator on ${\cal F}|_{n}^{n-1}$ (e.g., $\int_{\Lambda\gamma^{n-1}}^{\Lambda}dk\,\rho(k)b(k)\,\restrict{\cal F}_{n}\equiv\int_{\Lambda\gamma^{n-1}}^{\Lambda}dk\,\rho(k)b(k)\otimes\charf_{\mathcal{F}_{n}^{n-1}}$).We
adapt the notation for the Hamiltonians
\[
H_{P,n}:=H_{P}|_{\Lambda\gamma^{n}}^{\Lambda}=\sqrt{(P-P^{f})^{2}+m^{2}}+H^{f}+g\int_{\Lambda\gamma^{n}}^{\Lambda}dk\,\rho(k)\left(b(k)+b^{*}(k)\right)\,,
\]
and note 
\[
H_{P,n}=H_{P,n-1}+g\Phi|_{n}^{n-1},\qquad\Phi|_{n}^{n-1}:=\phi|_{n}^{n-1}+\phi^{*}|_{n}^{n-1},\qquad\phi|_{n}^{n-1}:=\int_{\Lambda\gamma^{n}}^{\Lambda\gamma^{n-1}}dk\,\rho(k)b(k)\,.
\]

\end{defn}
Furthermore, for simplicity of our presentation we keep an infrared
cut-off
\[
\kappa\equiv\Lambda\gamma^{N}=1\,,
\]
and in the following, for fixed $\Lambda$, the fineness parameter $\gamma$
will be chosen in such a way that

\begin{equation}
N=\frac{\ln\Lambda}{-\ln\gamma}\label{eq:num_steps}
\end{equation}
is an integer. Note that by construction $1\leq\Lambda\gamma^{n}\leq\Lambda$
for $0\leq n\leq N$. 
\begin{rem}
We warn the reader that, though it is not explicit in the notation,
the definitions of $\mathcal{F}_{n}$ and $H_{P,n}$ are $\Lambda-$dependent
as well as for other quantities introduced later on (e.g., $E_{P,n}$,
$\Psi_{P,n})$.
\end{rem}
We introduce:
\begin{defn}
For $P\in\mathbb{R}^{3}$ and integers $0\leq n\leq N$ we define
the ground state energies 
\[
E_{P,n}:=\inf\spec{H_{P,n}\restrict{\cal F}_{n}}.
\]
\end{defn}
The desired expansion formulas are a byproduct of the construction
of the ground states of the Hamiltonians $H_{P,N}\restrict{\cal F}_{N}$,
$|P|<P_{max}$. At the heart of this construction lies an induction
argument. Suppose that:
\begin{elabeling}{00.00.0000}
\item [{\emph{(i)}}] At step $(n-1)$ the vector $\Psi_{P,n-1}$ is the
unique ground state of the Hamiltonian $H_{P,n-1}\restrict{\cal F}_{n-1}$
with corresponding ground state energy $E_{P,n-1}$. 
\item [{\emph{(ii)}}] For some $\zeta>0$ the spectral gap can be bounded
from below by 
\[
\gap{H_{P,n-1}}{{\cal F}_{n-1}}\geq\zeta\omega\left(\Lambda\gamma^{n}\right).
\]

\end{elabeling}
Given the assumptions (i) and (ii) we can derive the implications
reported below.
\begin{enumerate}
\item In \lemref{gap_estimate} we show through a variational argument that
\[
\gap{H_{P,n-1}}{{\cal F}_{n}}\geq\zeta\omega\left(\Lambda\gamma^{n}\right).
\]

\item Next, we justify the Neumann expansion of the resolvent $\frac{1}{H_{P,n}-z}$
in terms of $\frac{1}{H_{P,n-1}-z}$ and the slice interaction $H_{P,n}-H_{P,n-1}$
for $z\in\mathbb{C}$ in the domain defined by 
\[
\frac{1}{2}\zeta\omega\left(\Lambda\gamma^{n+1}\right)\leq\left|E_{P,n-1}-z\right|\leq\zeta\omega\left(\Lambda\gamma^{n+1}\right)
\]
 by a direct computation; see \lemref{neumann}. We find
\[
\left\Vert \left(\frac{1}{H_{P,n-1}-z}\right)^{1/2}g\Phi|_{n}^{n-1}\left(\frac{1}{H_{P,n-1}-z}\right)^{1/2}\right\Vert _{{\cal {\cal F}}_{n}}\leq C|g|
\]
uniformly in $n$ and in $\Lambda$. The reason for this is that we
add interaction slices starting from $\Lambda$ down to $\Lambda\gamma^{N}=1$
in decreasing order so that the contribution of 
\[
\left\Vert g\phi|_{n}^{n-1}\left(\frac{1}{H_{P,n-1}-z}\right)^{1/2}\right\Vert _{{\cal F}_{n}}\leq C|g|\left(\Lambda\gamma^{n-1}(1-\gamma)\right)^{1/2}
\]
 is compensated thanks to the spectral gap estimate and the chosen
domain for $z$ which gives
\[
\left\Vert \left(\frac{1}{H_{P,n-1}-z}\right)^{1/2}\right\Vert _{{\cal F}_{n}}\leq C\left(\frac{1}{\Lambda\gamma^{n}(1-\gamma)}\right)^{1/2}.
\]

\item Finally, \thmref{mass_shell} ensures the existence of a unique ground
state 
\begin{multline}
\Psi_{P,n}:=-\frac{1}{2\pi i}\oint_{\Gamma_{P,n}}\frac{dz}{H_{P,n}-z}\Psi_{P,n-1}\\
=-\frac{1}{2\pi i}\sum_{j=0}^{\infty}\oint_{\Gamma_{P,n}}\frac{dz}{H_{P,n-1}-z}\left[-(H_{P,n}-H_{P,n-1})\frac{1}{H_{P,n-1}-z}\right]^{j}\Psi_{P,n-1}\label{eq:expansion formula-1}
\end{multline}
of the Hamiltonian $H_{P,n}\restrict{\cal F}_{n}$ by analytic perturbation
theory for sufficiently small $|g|$ uniformly in $n$ and $\Lambda<\infty$,
where the contour $\Gamma_{P,n}$ is appropriately chosen around $E_{P,n-1}$;
see \defref{ground_states}.
\item Furthermore, another variational argument guarantees $E_{P,n}\leq E_{P,n-1}$
and, hence, by Kato's theorem 
\[
\gap{H_{P,n}}{{\cal F}_{n}}\geq\zeta\omega\left(\Lambda\gamma^{n+1}\right).
\]

\end{enumerate}
Along this construction we gain the expansion formula (\ref{eq:expansion formula-1})
of the ground state $\Psi_{P,n}$ in terms of the previous ground
state $\Psi_{P,n-1}$ for each induction step. The above induction
is based on the following well-known results:
\begin{prop}
For $P\in\mathbb{R}^{3}$ and any integer $0\leq n<\infty$ the Hamiltonians
$H_{P}^{nuc},H^{f},H_{P}^{(0)},H_{P,n}$ acting on ${\cal F}$ are
essentially self-adjoint on the domain $D(H_{P=0}^{(0)})$ and bounded
from below.
\end{prop}

\begin{thm}
For $P\in\mathbb{R}^{3}$ and integers $0\leq n<\infty$ the ground
state energies fulfill
\begin{equation}
E_{P,n}\geq E_{0,n}.\label{eq:gross}
\end{equation}

\end{thm}
The inequality in (\ref{eq:gross}) is due to \cite{gross_existence_1972}.
\begin{rem}
We remark that the construction of the ground state can be implemented
for $\gamma$ arbitrarily close to $1$. This feature of our technique
will be crucial to derive the results on the limiting regime, as $\Lambda\to\infty$,
of the ground state energy and of the effective velocity stated in
Theorems (\ref{thm:energy_bounds}) and (\ref{thm:effective_velocity-2}),
respectively. Indeed, by (\ref{eq:num_steps}) it allows us to control
any error term that can be bounded by $\mathcal{O}(N(1-\gamma)^{1+\varepsilon})$
with $\varepsilon>0$.
\end{rem}
\textbf{Main Results.} As a direct application of the established
expansion formulas we can bound the ground state energy from above
and from below. The bounds are sharp in the sense that they identify
the order of dependence of the ground state energy on the ultraviolet
cut-off $\Lambda$ and the coupling constant $g$:
\begin{thm}
\label{thm:energy_bounds}Let $|g|$ be sufficiently small and $|P|<P_{max}$.
Define $E_{P,\Lambda}:=\inf\spec{H_{P}|_{\kappa}^{\Lambda}}$. There
exist universal constants $a,b>0$ such that for all $1<\Lambda<\infty$
it holds 
\begin{equation}
\sqrt{P^{2}+m^{2}}-g^{2}b\Lambda\leq E_{P,\Lambda}\leq\sqrt{P^{2}+m^{2}}-g^{2}a\Lambda\label{eq:ground_state_energy_estimate-1}
\end{equation}

\end{thm}
The proof will be given in the end of \secref{mass_shell}. \\

In our second main result we give an estimate of the effective velocity
of the nucleon in a one-particle state:
\begin{thm}
\label{thm:effective_velocity-2}Let $|g|$ be sufficiently small
and $|P|<P_{max}$. Then, there exist universal constants $c_{1},C_{1}>0$
such that the following estimate holds true
\end{thm}
\begin{equation}
\lim sup_{\gamma\to1}\left|\frac{\partial E_{P,N}}{\partial P_{i}}\right|\leq\Lambda^{-g^{2}c_{1}}\frac{|P|}{\left[P^{2}+m^{2}\right]^{1/2}}+C_{1}|g|^{1/2},\qquad i=1,2,3.\label{thm:effective_velocity-1-1}
\end{equation}
The proof will be given in Section (\ref{sec:eff_vel_eff_mass}).
A direct consequence of the bound in (\ref{thm:effective_velocity-1-1})
is

\begin{equation}
\lim sup_{\Lambda\to\infty}\left|\frac{\partial E_{P,\Lambda}}{\partial P_{i}}\right|\leq C|g|^{1/2}.\label{eq:eff_mass_limit-1}
\end{equation}
In order to interpret this result consider that in the free case,
i.e., $g=0$, one finds 
\[
\left|\frac{\partial E_{P,\Lambda}}{\partial P_{i}}\right|=\frac{|P_{i}|}{\sqrt{P^{2}+m^{2}}}.
\]
Therefore, \thmref{effective_velocity-2} states that if the interaction
is turned on, even for an arbitrarily small but non-zero $|g|$,
 the absolute value of the gradient
of the ground state energy decreases to an order smaller or equal
to $|g|^{1/2}$ in the limit $\Lambda\to\infty$. The physical interpretation of this result is that
the mass shell essentially becomes flat and the theory trivial in
the limit $\Lambda\to\infty$. Moreover, our proof shows that not
even a suitable scaling of the bare mass, i.e., $m\equiv m(\Lambda)>0$,
may prevent the mass shell from becoming essentially flat. 

A crucial tool for the above results comes from the non-perturbative
estimates that we derive in Theorem (\ref{thm:expansion_formulas})
and Theorem (\ref{lem:ground_state_energy_bounds}), respectively:
\begin{equation}
a\Lambda\gamma^{n-1}(1-\gamma)\leq\left\langle \widehat{\Psi}_{P,n-1},\phi|_{n}^{n-1}\frac{1}{H_{P,n-1}-E_{P,n-1}}\phi^{*}|_{n}^{n-1}\widehat{\Psi}_{P,n-1}\right\rangle \leq b\Lambda\gamma^{n-1}(1-\gamma)\,,\label{eq:tool-1}
\end{equation}
\begin{equation}
c_{1}(1-\gamma)\leq\alpha_{P}|_{n}^{n-1}:=\left\langle \widehat{\Psi}_{P,n-1},\phi|_{n}^{n-1}\left(\frac{1}{H_{P,n-1}-E_{P,n-1}}\right)^{2}\phi^{*}|_{n}^{n-1}\widehat{\Psi}_{P,n-1}\right\rangle \leq c_{2}(1-\gamma)\label{eq:tool-2}
\end{equation}
which hold for some universal constants $0<a\leq b<\infty$, $0<c_{1}\leq c_{2}<\infty$.
In order to get the bounds in (\ref{eq:tool-1})-(\ref{eq:tool-2})
we make use of the spectral information obtained during the construction
of the ground states.

The strategy of proof in \thmref{effective_velocity-2} consists in
re-expanding back the vectors in the matrix element yielding the effective
velocity. This means that, iteratively, the matrix element 
\[
\left\langle \widehat{\Psi}_{P,n},V_{i}(P)\widehat{\Psi}_{P,n}\right\rangle \equiv\frac{\partial E_{P,n}}{\partial P_{i}},\qquad V_{i}(P):=\frac{P_{i}-P_{i}^{f}}{\left[(P-P^{f})^{2}+m^{2}\right]^{1/2}}
\]
will be expressed in terms of:
\begin{enumerate}
\item The analogous quantity on scale $n-1$, i.e., 
\begin{equation}
\Big\langle\widehat{\Psi}_{P,n-1},V_{i}(P)\widehat{\Psi}_{P,n-1}\Big\rangle\label{eq:leading term}
\end{equation}

\item The scalar products
\[
A_{P,n-1}:=g^{2}\left\langle \frac{1}{H_{P,n-1}-E_{P,n-1}}\phi^{*}|_{n}^{n-1}\widehat{\Psi}_{P,n-1},V_{i}(P)\frac{1}{H_{P,n-1}-E_{P,n-1}}\phi^{*}|_{n}^{n-1}\widehat{\Psi}_{P,n-1}\right\rangle 
\]
and 
\[
B_{P,n-1}:=2g^{2}\Re\left\langle {\cal \widetilde{{\cal Q}}}_{P,n-1}^{\perp}\frac{1}{H_{P,n-1}-E_{P,n-1}}\phi|_{n}^{n-1}\frac{1}{H_{P,n-1}-E_{P,n-1}}\phi^{*}|_{n}^{n-1}\widehat{\Psi}_{P,n-1},V_{i}(P)\widehat{\Psi}_{P,n-1}\right\rangle 
\]
where ${\cal \widetilde{{\cal Q}}}_{P,n-1}^{\perp}$ is defined in
\eqref{Qtilde_def}.\\

\item A remainder that can be estimated to be $\mathcal{O}(|g|^{4}(1-\gamma)^{\frac{4}{2}})$.
\end{enumerate}
The hard part of our proof is showing that some a priori estimates
on $A_{P,n-1}$ and $B_{P,n-1}$ hold so that they shall not be re-expanded
like the leading term (\ref{eq:leading term}) but their cumulative
contribution can be estimated to be of order $|g|^{\frac{1}{2}}$
as in (\ref{thm:effective_velocity-1-1}). Two substantially different
arguments are devised to control $A_{P,n-1}$ and $B_{P,n-1}$:
\begin{itemize}
\item As for $A_{P,n-1}$, due to the velocity operator $V_{i}(P)$ we can
show summability in $n$ after contracting the boson operators $\phi^{*}|_{n}^{n-1}$.
\item As for $B_{P,n-1}$, by exploiting the presence of the orthogonal
projection ${\cal \widetilde{{\cal Q}}}_{P,n-1}^{\perp}$ and a suitable
\emph{one-step, $g-$dependent} backwards expansion\emph{,} we can improve the
crude estimate, $\mathcal{O}(g^{2}(1-\gamma))$, that follows from
the operator bounds derived in Section \ref{sec:mass_shell} by, at
least, an extra factor $|g|^{\frac{1}{2}}$.
\end{itemize}
The product of the coefficients $\{(1-g^{2}\alpha_{P}|_{n}^{n-1})\}_{1\leq n\leq N}$
that are generated in front of the leading term (\ref{eq:leading term})
at each step of the re-expansion gives rise to a damping factor bounded
above by $\Lambda^{-g^{2}c_{1}}$ as $\gamma$ tends to $1$.

\section{\label{sec:mass_shell}Construction of the One-Particle States}

We begin our discussion with the construction of the ground states
corresponding to the Hamiltonians $H_{P,n}\restrict{\cal F}_{n}$,
$0\leq n\leq N$. This construction is based on an induction completed
in \thmref{mass_shell}. Next, we collect helpful estimates and expansion
formulas which also will be used frequently in \secref{eff_vel_eff_mass}.
This section ends with \lemref{ground_state_energy_bounds} where
we derive some upper and lower bounds on the ground state energies. 

The first lemma provides some a priori estimates on the ground state
energies. In particular claim (iii) of \lemref{energy} will be crucial
for the gap estimate in \lemref{gap_estimate}.
\begin{lem}
\label{lem:energy}For $P\in\mathbb{R}^{3}$ and any integer $0\leq n<N$
suppose $\Psi_{P,n}$ is the ground state of $H_{P,n}\restrict{\cal F}_{n}$
and $E_{P,n}$ is the corresponding ground state energy. Then:
\begin{enumerate}[label=(\roman*)]
\item  $E_{P,n+1}\leq E_{P,n}$.
\item $-g^{2}C\Lambda\leq E_{P,n}\leq\sqrt{P^{2}+m^{2}}$.
\item $\forall k\in\mathbb{\mathbb{R}}^{3}\,,\,\, E_{P-k,n}-E_{P,n}\geq-|P|\omega(k)$.
\end{enumerate}
\end{lem}
\begin{proof}
~
\begin{enumerate}[label=(\roman*)]
\item By definition of the ground state energy we can estimate
\[
E_{P,n+1}-E_{P,n}\leq\frac{\left\langle \Psi_{P,n},\left[H_{P,n+1}-H_{P,n}\right]\Psi_{P,n}\right\rangle }{\left\langle \Psi_{P,n},\Psi_{P,n}\right\rangle }=\frac{\left\langle \Psi_{P,n},\, g\Phi|_{n+1}^{n}\Psi_{P,n}\right\rangle }{\left\langle \Psi_{P,n},\Psi_{P,n}\right\rangle }=0.
\]

\item It suffices to observe that
\[
E_{P,n}\leq\left\langle \widehat{\Psi}_{P,0},H_{P,n}\widehat{\Psi}_{P,0}\right\rangle =\sqrt{P^{2}+m^{2}}
\]
and
\[
0\leq\sqrt{(P-P^{f})^{2}+m^{2}}+\int_{\Lambda\gamma^{n}}^{\Lambda}dk\,\omega(k)\left(b_{k}^{*}+g\frac{\rho(k)}{\omega(k)}\right)\left(b_{k}+g\frac{\rho(k)}{\omega(k)}\right)=H_{P,n}+g^{2}\int_{\Lambda\gamma^{n}}^{\Lambda}dk\,\frac{\rho(k)^{2}}{\omega(k)}
\]
where
\[
g^{2}\int_{\Lambda\gamma^{n}}^{\Lambda}dk\,\frac{\rho(k)^{2}}{\omega(k)}\leq g^{2}C\Lambda.
\]

\item Inequality (\ref{eq:gross}) implies
\[
E_{P-k,n}-E_{P,n}=E_{P-k,n}-E_{0,n}+E_{0,n}-E_{P,n}\geq E_{0,n}-E_{P,n}
\]
and
\begin{multline}
E_{0,n}-E_{P,n}\geq\frac{\left\langle \Psi_{0,n},\left[H_{0,n}-H_{P,n}\right]\Psi_{0,n}\right\rangle }{\left\langle \Psi_{0,n},\Psi_{0,n}\right\rangle }=\frac{\left\langle \Psi_{0,n},\left[H_{0}^{nuc}-H_{P}^{nuc}\right]\Psi_{0,n}\right\rangle }{\left\langle \Psi_{0,n},\Psi_{0,n}\right\rangle }\geq-|P|\geq-|P|\omega(k)\\
\label{eq:energy_difference}
\end{multline}
because 
\[
\left\|\sqrt{P^{f}{}^{2}+m^{2}}-\sqrt{(P-P^{f})^{2}+m^{2}}\right\|\leq|P|
\]
and $\omega(k)=\sqrt{k^{2}+\mu^{2}}$ with $\mu>1$.
\end{enumerate}
\end{proof}

In our construction we shall single out two parameters needed to control
the gap of the Hamiltonians $H_{P,n}\restrict{\cal F}_{n}$, $0\leq n\leq N$:
\begin{defn}
\label{def:gap_params}Define $\frac{1}{8}<\theta<\frac{1}{4}$ and
$\zeta>\frac{1}{4}$ such that
\[
1-\theta-\pmax\geq\zeta.
\]

\end{defn}
Later the following lemma will be invoked from the main induction
in \thmref{mass_shell} to provide the gap estimate that is used in
the inductive scheme.
\begin{lem}
\label{lem:gap_estimate}Let $|P|<P_{max}$ and $1\leq n\leq N$.
Assume:
\begin{enumerate}[label=A(\roman*)]
\item $E_{P,n-1}$ is the non-degenerate ground state energy of $H_{P,n-1}\restrict{\cal F}_{n-1}$
corresponding to the ground state vector $\Psi_{P,n-1}.$
\item $\gap{H_{P,n-1}}{{\cal F}_{n-1}}\geq\zeta\omega\left(\Lambda\gamma^{n}\right)$.
\end{enumerate}
Then:
\begin{enumerate}[label=C(\roman*)]
\item $E_{P,n-1}$ is the non-degenerate ground state energy of $H_{P,n-1}\restrict{\cal F}_{n}$
corresponding to the ground state vector $\Psi_{P,n-1}\otimes\Omega.$
\item 
\[
\gap{H_{P,n-1}}{{\cal F}_{n}},\,\inf_{\varphi=\psi\otimes\eta}\left\langle \widehat{\varphi},\left(H_{P,n-1}-\theta H^{f}|_{n}^{n-1}-E_{P,n-1}\right)\widehat{\varphi}\right\rangle \geq\zeta\omega\left(\Lambda\gamma^{n}\right)
\]
where the infimum is taken over $\varphi\in D(H_{P}^{(0)})$ such
that $\psi\in{\cal F}_{n-1}$ and $\eta\in{\cal F}|_{n}^{n-1}$ contains
a strictly positive number of bosons.
\end{enumerate}
\end{lem}
\begin{proof}
A direct computation using A(i) shows that $\Psi_{P,n-1}\otimes\Omega$
is eigenvector of $H_{P,n-1}\restrict{\cal F}_{n}$ with corresponding
eigenvalue $E_{P,n-1}$. Since $H^{f}|_{n}^{n-1}$ is a positive operator
one has 
\begin{equation}
\inf_{\varphi\perp\Psi_{P,n-1}\otimes\Omega}\left\langle \widehat{\varphi},\left(H_{P,n-1}-E_{P,n-1}\right)\widehat{\varphi}\right\rangle \geq\inf_{\varphi\perp\Psi_{P,n-1}\otimes\Omega}\left\langle \widehat{\varphi},\left(H_{P,n-1}-\theta H^{f}|_{n}^{n-1}-E_{P,n-1}\right)\widehat{\varphi}\right\rangle \,;\label{eq:lemma-gap inf}
\end{equation}
we subtract the term $\theta H^{f}|_{n}^{n-1}$ for a technical reason
which will become clear in Lemma \ref{lem:neumann}.

Now, the right-hand side of (\ref{eq:lemma-gap inf}) is bounded from
below by
\begin{equation}
\min\left\{ \gap{H_{P,n-1}}{{\cal F}_{n-1}},\inf_{\varphi=\psi\otimes\eta}\left\langle \widehat{\varphi},\left(H_{P,n-1}-\theta H^{f}|_{n}^{n-1}-E_{P,n-1}\right)\widehat{\varphi}\right\rangle \right\} ,\label{eq:gap_min-1-1}
\end{equation}
where $\psi\in\mathcal{F}_{n-1},\eta\in\mathcal{F}_{n}^{n-1},\psi\otimes\eta$
belongs to $D(H_{P}^{(0)})$, and $\eta$ is a vector with definite,
strictly positive number of bosons. For a vector $\eta$ with $l\geq1$
bosons we compute 
\begin{multline*}
\inf_{\varphi=\psi\otimes\eta}\left\langle \widehat{\varphi},\left(H_{P,n-1}-\theta H^{f}|_{n}^{n-1}-E_{P,n-1}\right)\widehat{\varphi}\right\rangle \\
\geq\inf_{\psi,\Lambda\gamma^{n}\leq\left|k_{j}\right|\leq\Lambda\gamma^{n-1}}\left\langle \widehat{\psi},\left(H_{P-\sum_{j=1}^{l}k_{j},n-1}+(1-\theta)\sum_{j=1}^{l}\omega(k_{j})-E_{P,n-1}\right)\widehat{\psi}\right\rangle \\
\geq\inf_{\psi,\Lambda\gamma^{n}\leq\left|k_{j}\right|\leq\Lambda\gamma^{n-1}}\left(E_{P-\sum_{j=1}^{l}k_{j},n-1}-E_{P,n-1}+(1-\theta)\sum_{j=1}^{l}\omega(k_{j})\right).
\end{multline*}
Furthermore, \lemref{energy} implies
\[
E_{P-\sum_{j=1}^{l}k_{j},n-1}-E_{P,n-1}\geq-\pmax\sum_{j=1}^{l}\omega(k_{j}).
\]
Hence, by \defref{gap_params} the inequality
\[
\inf_{\varphi=\psi\otimes\eta}\left\langle \widehat{\varphi},\left(H_{P,n-1}-\theta H^{f}|_{n}^{n-1}-E_{P,n-1}\right)\widehat{\varphi}\right\rangle \geq\zeta\omega\left(\Lambda\gamma^{n}\right)
\]
 holds. Now by A(ii) we also get

\begin{equation}
(\ref{eq:gap_min-1-1})\geq\zeta\omega\left(\Lambda\gamma^{n}\right).\label{eq:gap_min-1}
\end{equation}
From the estimate in \eqref{gap_min-1} we can conclude that $\Psi_{P,n-1}\otimes\Omega$
is the unique ground state of $H_{P,n-1}\restrict{\cal F}_{n}$ with
eigenvalue $E_{P,n-1}$ and 
\[
\gap{H_{P,n-1}}{{\cal F}_{n}}\geq\zeta\omega\left(\Lambda\gamma^{n}\right).
\]
This proves C(i) and C(ii). 
\end{proof}
The second ingredient needed for the main induction in \thmref{mass_shell}
is a control of the resolvent expansion of the Hamiltonians:
\begin{lem}
\label{lem:neumann}Let $|g|$ be sufficiently small and $|P|<P_{max}$.
Suppose further that for $1\leq n\leq N$ $E_{P,n-1}$ is the non-degenerate
ground state energy of $H_{P,n-1}\restrict{\cal F}_{n-1}$ corresponding
to the ground state vector $\Psi_{P,n-1}$ and that
\begin{equation}
\gap{H_{P,n-1}}{{\cal F}_{n}}\geq\zeta\omega\left(\Lambda\gamma^{n}\right).\label{eq:gap_assumption}
\end{equation}
Then, for $z\in\mathbb{C}$ such that
\[
\frac{1}{2}\zeta\omega\left(\Lambda\gamma^{n+1}\right)\leq\left|E_{P,n-1}-z\right|\leq\zeta\omega\left(\Lambda\gamma^{n+1}\right),
\]
the resolvent $\frac{1}{H_{P,n}-z}$ is a well-defined operator on
${\cal F}_{n}$ which equals to
\begin{equation}
\frac{1}{H_{P,n-1}-z}\sum_{j=0}^{\infty}\left[-g\Phi|_{n}^{n-1}\frac{1}{H_{P,n-1}-z}\right]^{j}.\label{eq:powerseries}
\end{equation}

\begin{proof}
We start with the estimate
\begin{multline*}
\left\Vert \left(\frac{1}{H_{P,n-1}-z}\right)^{1/2}\right\Vert _{{\cal F}_{n}}=\frac{1}{\sqrt{\mathrm{dist}\left(z,\spec{H_{P,n-1}\restrict{\cal F}_{n}}\right)}}\\
\leq\left(\max\left\{ \frac{2}{\zeta\omega\left(\Lambda\gamma^{n+1}\right)},\frac{C}{\zeta\omega\left(\Lambda\gamma^{n}\right)-\zeta\omega\left(\Lambda\gamma^{n+1}\right)}\right\} \right)^{1/2}\leq\left(\frac{C}{\zeta\Lambda\gamma^{n+1}(1-\gamma)}\right)^{1/2}
\end{multline*}
where we made use of the assumption in (\ref{eq:gap_assumption}).
Next, we estimate
\begin{equation}
\left\Vert g\phi|_{n}^{n-1}\left(\frac{1}{H_{P,n-1}-z}\right)^{1/2}\right\Vert _{{\cal F}_{n}}\leq|g|C\left[\Lambda\gamma^{n-1}(1-\gamma)\right]^{1/2}\left\Vert \left(H^{f}|_{n}^{n-1}\right)^{1/2}\left(\frac{1}{H_{P,n-1}-z}\right)^{1/2}\right\Vert _{{\cal F}_{n}}.\label{eq:phi_res_est}
\end{equation}
The operators $H^{f}|_{n}^{n-1}$ and $H_{P,n-1}$ commute, and we
may apply the spectral theorem and \lemref{gap_estimate} in order
to get
\[
\left\Vert \left(H^{f}|_{n}^{n-1}\right)^{1/2}\left(\frac{1}{H_{P,n-1}-z}\right)^{1/2}\right\Vert _{{\cal F}_{n}}=\left\Vert \left(H^{f}|_{n}^{n-1}\right)^{1/2}\left(\frac{1}{H_{P,n-1}-\theta H^{f}|_{n}^{n-1}-z+\theta H^{f}|_{n}^{n-1}}\right)^{1/2}\right\Vert _{{\cal F}_{n}}\leq\theta^{-1/2}.
\]
In consequence, we can estimate
\[
\left\Vert g\left(\frac{1}{H_{P,n-1}-z}\right)^{1/2}\Phi|_{n}^{n-1}\left(\frac{1}{H_{P,n-1}-z}\right)^{1/2}\right\Vert _{{\cal F}_{n}}\leq|g|C(\zeta\gamma^{2})^{-1/2}\theta^{-1/2}.
\]
Since $\gamma>\frac{1}{2}$, $\zeta>\frac{1}{4}$, and $\theta>\frac{1}{8}$
the coupling constant $|g|$ can be chosen independently of $n$ (and
of $\Lambda$) such that
\[
|g|C(\theta\zeta\gamma^{2})^{-1/2}<1
\]
which implies the convergence of the power series on the right-hand
side of (\ref{eq:powerseries}) and, thus, the claim.
\end{proof}
\end{lem}
We will now prove that the vectors in the following definition are
the unique, non-zero ground states of the Hamiltonians $H_{P,n}\restrict{\cal F}_{n}$
, $0\leq n\leq N$. (We warn the reader that the spectral projection
in (\ref{eq:projection/contour}) will be shown to be well defined
in Theorem \ref{thm:mass_shell}.)
\begin{defn}
\label{def:ground_states}For $1\leq n\leq N$ we define
\begin{equation}
{\cal Q}_{P,n}:=-\frac{1}{2\pi i}\oint_{\Gamma_{P,n}}\frac{dz}{H_{P,n}-z\,}\restrict{\cal F}_{n}\qquad\Gamma_{P,n}:=\left\{ z\in\mathbb{C}\,\bigg|\,\left|E_{P,n-1}-z\right|=\frac{1}{2}\zeta\omega\left(\Lambda\gamma^{n+1}\right)\right\} \label{eq:projection/contour}
\end{equation}
and recursively
\begin{equation}
\Psi_{P,n}:={\cal Q}_{P,n}\Psi_{P,n-1},\qquad\Psi_{P,0}:=\Omega.\label{vectors by iteration}
\end{equation}
Note that $\Psi_{P,n}$ are in general unnormalized vectors with $\|\Psi_{P,n}\|\leq1$.\end{defn}
\begin{thm}
\label{thm:mass_shell}Let $|g|$ be sufficiently small and $|P|<P_{max}$.
For $0\leq n\leq N$ it holds:
\begin{enumerate}[label=(\roman*)]
\item $\Psi_{P,n}$ is well-defined, non-zero, and the unique ground state
vector of $H_{P,n}\restrict{\cal F}_{n}$ with corresponding eigenvalue
\[
E_{P,n}:=\inf\spec{H_{P,n}\restrict{\cal F}_{n}}.
\]

\item $\gap{H_{P,n}}{{\cal F}_{n}}\geq\zeta\omega\left(\Lambda\gamma^{n+1}\right).$
\end{enumerate}
\end{thm}
\begin{proof}
A direct computation shows that the claim holds for $n=0$. Let us
assume it holds for $n-1$ with $0\leq n-1<N-1$:
\begin{enumerate}
\item The assumptions allow to apply \lemref{gap_estimate} which states
that
\[
\gap{H_{P,n-1}}{{\cal F}_{n}}\geq\zeta\omega\left(\Lambda\gamma^{n}\right).
\]

\item Hence, \lemref{neumann} ensures that for $|g|$ small enough but
uniform in $n$ (and in $\Lambda$) the resolvent
\[
\frac{1}{H_{P,n}-z}\restrict{\cal F}_{n}=\frac{1}{H_{P,n-1}-z}\sum_{j=0}^{\infty}\left[-g\Phi|_{n}^{n-1}\frac{1}{H_{P,n-1}-z}\right]^{j}\restrict{\cal F}_{n}
\]
is well-defined for 
\begin{equation}
\frac{1}{2}\zeta\omega\left(\Lambda\gamma^{n+1}\right)\leq\left|E_{P,n-1}-z\right|\leq\zeta\omega\left(\Lambda\gamma^{n+1}\right).\label{eq:perturb_radius}
\end{equation}

\item For $|g|$ small enough but uniform in $n$ (and in $\Lambda$), $\Psi_{P,n}$
defined in (\ref{vectors by iteration}) is non-zero. Indeed for $0\leq n\leq N$ and $z\in\Gamma_{P,n}$ we have
\[
\left\Vert \left(\frac{1}{H_{P,n-1}-z}\right)^{1/2}g\Phi|_{n}^{n-1}\left(\frac{1}{H_{P,n-1}-z}\right)^{1/2}\right\Vert _{{\cal F}_{n}}\leq C|g|(1-\gamma)^{1/2}
\]
because for $z$ in the domain $\Gamma_{P,n}$ defined
in (\ref{eq:projection/contour}) we get
\begin{equation}
\left\Vert \left(\frac{1}{H_{P,n-1}-z}\right)^{1/2}\right\Vert _{{\cal F}_{n}}\leq\left(\frac{C}{\Lambda\gamma^{n}}\right)^{1/2}\nonumber\label{eq:estimate norm contour}
\end{equation}
that we can combine with
the bound in (\ref{eq:phi_res_est}).
By Kato's theorem
we can conclude that it is the unique ground state of $H_{P,n}\restrict{\cal F}_{n}$
with corresponding ground state energy $E_{P,n}$.
\item \lemref{energy}(i), Kato's theorem, and the domain of $z$ given
in (\ref{eq:perturb_radius}) provide the estimate
\[
\gap{H_{P,n}}{{\cal F}_{n}}\geq\zeta\omega\left(\Lambda\gamma^{n+1}\right).
\]

\end{enumerate}
\end{proof}
Next we provide expansion formulas which will be used frequently in
our computations in \secref{eff_vel_eff_mass}.
\begin{thm}
\label{thm:expansion_formulas}Let $|g|$ be sufficiently small and
$|P|<P_{max}$. For $0\leq n\leq N$ the following statements hold:
\begin{enumerate}[label=(\roman*)]
\item The following equality is satisfied: 
\begin{align*}
\Psi_{P,n}= & \Psi_{P,n-1}-g\frac{1}{H_{P,n-1}-E_{P,n-1}}\phi^{*}|_{n}^{n-1}\Psi_{P,n-1}\\
 & +g^{2}\widetilde{{\cal Q}}_{P,n-1}^{\perp}\frac{1}{H_{P,n-1}-E_{P,n-1}}\phi|_{n}^{n-1}\frac{1}{H_{P,n-1}-E_{P,n-1}}\phi^{*}|_{n}^{n-1}\Psi_{P,n-1}\\
 & +g^{2}\widetilde{{\cal Q}}_{P,n-1}^{\perp}\frac{1}{H_{P,n-1}-E_{P,n-1}}\phi^{*}|_{n}^{n-1}\frac{1}{H_{P,n-1}-E_{P,n-1}}\phi^{*}|_{n}^{n-1}\Psi_{P,n-1}\\
 & -g^{2}\widetilde{{\cal Q}}_{P,n-1}\phi|_{n}^{n-1}\left(\frac{1}{H_{P,n-1}-E_{P,n-1}}\right)^{2}\phi^{*}|_{n}^{n-1}\Psi_{P,n-1}+{\cal O}\left(|g|^{3}(1-\gamma)^{3/2}\right)
\end{align*}
for
\begin{equation}
\widetilde{{\cal Q}}_{P,n-1}:=-\frac{1}{2\pi i}\oint_{\Gamma_{P,n}}dz\frac{1}{H_{P,n-1}-z}\restrict{\cal F}_{n},\qquad\widetilde{{\cal Q}}_{P,n-1}^{\perp}:=\charf_{{\cal F}_{n}}-\widetilde{{\cal Q}}_{P,n-1}\label{eq:Qtilde_def}
\end{equation}
where $\charf_{{\cal F}_{n}}$ is the identity operator on ${\cal F}_{n}$.\\

\item The norm of the ground state vectors fulfills the relation
\begin{equation}
\left\Vert \Psi_{P,n}\right\Vert ^{2}=\left\langle \Psi_{P,n},\Psi_{P,n}\right\rangle =\left(1-g^{2}\alpha_{P}|_{n}^{n-1}+{\cal O}\left(|g|^{4}(1-\gamma)^{4/2}\right)\right)\left\Vert \Psi_{P,n-1}\right\Vert ^{2}\label{eq:norm_alpha}
\end{equation}
where
\[
\alpha_{P}|_{n}^{n-1}:=\left\langle \widehat{\Psi}_{P,n-1},\phi|_{n}^{n-1}\left(\frac{1}{H_{P,n-1}-E_{P,n-1}}\right)^{2}\phi^{*}|_{n}^{n-1}\widehat{\Psi}_{P,n-1}\right\rangle .
\]

\item There exist universal constants $0<c_{1}\leq c_{2}<\infty$ such that
\[
c_{1}(1-\gamma)\leq\alpha_{P}|_{n}^{n-1}\leq c_{2}(1-\gamma).
\]

\end{enumerate}
\end{thm}
\begin{proof}
Claim (i) can be shown by a direct computation using \defref{ground_states}.
Likewise claim (ii) follows from \defref{ground_states} by exploiting
the relation
\[
\Psi_{P,n}={\cal Q}_{P,n}\Psi_{P,n-1}=\frac{\left\langle \Psi_{P,n},\Psi_{P,n-1}\right\rangle }{\left\langle \Psi_{P,n},\Psi_{P,n}\right\rangle }\Psi_{P,n}
\]
that holds by construction.\\

Next, we prove claim (iii). The bound from above is obtained by using
the pull-through formula and Lemma \ref{lem:energy} (iii), i.e.,
\begin{multline}
\alpha_{P}|_{n}^{n-1}=\left\langle \widehat{\Psi}_{P,n-1},\phi|_{n}^{n-1}\left(\frac{1}{H_{P,n-1}-E_{P,n-1}}\right)^{2}\phi^{*}|_{n}^{n-1}\widehat{\Psi}_{P,n-1}\right\rangle \\
=\int_{\Lambda\gamma^{n}}^{\Lambda\gamma^{n-1}}dk\,\rho(k)^{2}\left\langle \widehat{\Psi}_{P,n-1},\left(\frac{1}{H_{P-k,n-1}+\omega(k)-E_{P,n-1}}\right)^{2}\widehat{\Psi}_{P,n-1}\right\rangle \leq-C\ln\gamma\leq c_{2}(1-\gamma)\label{eq:alpha_bound_above}
\end{multline}
for an appropriately chosen constant $c_{2}$; recall that $\frac{1}{2}<\gamma<1.$ 

With respect to the bound from below we consider the spectral representation
for the self-adjoint operator $H_{P-k,n-1}+\omega(k)-E_{P,n-1}$ and
define the spectral projections
\[
\chi^{+}(k):=\chi_{(5\omega(k),+\infty)}\left(H_{P-k,n-1}+\omega(k)-E_{P,n-1}\right),\qquad\chi^{-}(q):\mathbb{=}\charf_{{\cal F}_{n-1}}-\chi^{+}(q)
\]
where $\chi_{(5\omega(k),+\infty)}$ is the characteristic function
being one on the interval $(5\omega(k),+\infty)$ and zero otherwise.
We also define the function
\[
f(k):=\rho(k)^{2}\left\langle \widehat{\Psi}_{P,n-1},\left(\frac{1}{H_{P-k,n-1}+\omega(k)-E_{P,n-1}}\right)^{2}\left(\chi^{+}(k)+\chi^{-}(k)\right)\widehat{\Psi}_{P,n-1}\right\rangle 
\]
that we study for two complementary cases:
\begin{enumerate}[label=(\alph*)]
\item In the case $\left\Vert \chi^{+}(k)\widehat{\Psi}_{P,n-1}\right\Vert ^{2}<\frac{1}{2}$
we get
\begin{equation}
f(k)\geq\rho(k)^{2}\left\langle \chi^{-}(k)\widehat{\Psi}_{P,n-1},\left(\frac{1}{H_{P-k,n-1}+\omega(k)-E_{P,n-1}}\right)^{2}\chi^{-}(k)\widehat{\Psi}_{P,n-1}\right\rangle \geq\frac{\rho(k)^{2}}{50\omega(k)^{2}}.\label{eq:fk_1}
\end{equation}

\item In the other case, i.e., $\left\Vert \chi^{+}(k)\widehat{\Psi}_{P,n-1}\right\Vert ^{2}\geq\frac{1}{2}$,
we start with the trivial inequality
\begin{equation}
f(k)\geq\rho(k)^{2}\left\langle \frac{1}{H_{P-k,n-1}+\omega(k)-E_{P,n-1}}\widehat{\Psi}_{P,n-1},\chi^{+}(k)\frac{1}{H_{P-k,n-1}+\omega(k)-E_{P,n-1}}\widehat{\Psi}_{P,n-1}\right\rangle \label{eq:fk_second}
\end{equation}
and consider the resolvent formulas 
\begin{multline}
\frac{1}{H_{P-k,n-1}+\omega(k)-E_{P,n-1}}=\frac{1}{H_{P,n-1}+\omega(k)-E_{P,n-1}}\\
-\frac{1}{H_{P,n-1}+\omega(k)-E_{P,n-1}}\Delta_{P}(k)\frac{1}{H_{P-k,n-1}+\omega(k)-E_{P,n-1}}\label{eq:resolvent formula 1}
\end{multline}
and 
\begin{multline}
\quad\frac{1}{H_{P-k,n-1}+\omega(k)-E_{P,n-1}}=\frac{1}{H_{P,n-1}+\omega(k)-E_{P,n-1}}\\
-\frac{1}{H_{P-k,n-1}+\omega(k)-E_{P,n-1}}\Delta_{P}(k)\frac{1}{H_{P,n-1}+\omega(k)-E_{P,n-1}}\label{eq:resolvent formula 2}
\end{multline}
where 
\[
\Delta_{P}(k):=\sqrt{(P-k-P^{f})^{2}+m^{2}}-\sqrt{(P-P^{f})^{2}+m^{2}}\,.
\]
Then we apply the expansions in (\ref{eq:resolvent formula 1}) and
in (\ref{eq:resolvent formula 2}) to the resolvents on the left and
on the right in the scalar product of (\ref{eq:fk_second}), respectively,
and get
\begin{multline}
f(k)\geq\rho(k)^{2}\left\langle \widehat{\Psi}_{P,n-1},\frac{1}{H_{P,n-1}+\omega(k)-E_{P,n-1}}\chi^{+}(k)\frac{1}{H_{P,n-1}+\omega(k)-E_{P,n-1}}\widehat{\Psi}_{P,n-1}\right\rangle \\
-2\Re\rho(k)^{2}\Bigg\langle\widehat{\Psi}_{P,n-1},\frac{1}{H_{P,n-1}+\omega(k)-E_{P,n-1}}\Delta_{P}(k)\frac{1}{H_{P-k,n-1}+\omega(k)-E_{P,n-1}}\times\\
\times\chi^{+}(k)\frac{1}{H_{P,n-1}+\omega(k)-E_{P,n-1}}\widehat{\Psi}_{P,n-1}\Bigg\rangle\\
+\rho(k)^{2}\left\langle \widehat{\Psi}_{P,n-1},\left|\chi^{+}(k)\frac{1}{H_{P-k,n-1}+\omega(k)-E_{P,n-1}}\Delta_{P}(k)\frac{1}{H_{P,n-1}+\omega(k)-E_{P,n-1}}\right|^{2}\widehat{\Psi}_{P,n-1}\right\rangle .\label{eq:f_binomial_form}
\end{multline}
Note that 
\[
\left\Vert \Delta_{P}(k)\right\Vert \leq|k|
\]
so that neglecting the last positive term in (\ref{eq:f_binomial_form})
we get the estimate
\begin{multline}
f(k)\geq\frac{\rho(k)^{2}}{\omega(k)^{2}}\left\Vert \chi^{+}(k)\widehat{\Psi}_{P,n-1}\right\Vert ^{2}-\frac{2\rho(k)^{2}|k|}{5\omega(k)^{3}}\left\Vert \chi^{+}(k)\widehat{\Psi}_{P,n-1}\right\Vert \\
\geq\frac{\rho(k)^{2}}{\omega(k)^{2}}\left\Vert \chi^{+}(k)\widehat{\Psi}_{P,n-1}\right\Vert \left(\frac{1}{\sqrt{2}}-\frac{2}{5}\right)\geq\frac{(5-2\sqrt{2})\rho(k)^{2}}{10\omega(k)^{2}}.\label{eq:fk_2-2}
\end{multline}

\end{enumerate}
Combining the bounds (\ref{eq:fk_1}) and (\ref{eq:fk_2-2}) we obtain
\[
\int_{\Lambda\gamma^{n}}^{\Lambda\gamma^{n-1}}dk\,\rho(k)^{2}\left\langle \widehat{\Psi}_{P,n-1},\left(\frac{1}{H_{P-k,n-1}+\omega(k)-E_{P,n-1}}\right)^{2}\widehat{\Psi}_{P,n-1}\right\rangle \geq-C\ln\gamma\geq c_{1}(1-\gamma)
\]
that gives the bound from below on $\alpha_{P}|_{n}^{n-1}$ for an
appropriately chosen constant $c_{1}$. This together with the bound
from above (\ref{eq:alpha_bound_above}) proves the claim.
\end{proof}
With the help of these expansion formulas we get upper and lower bounds
on the ground state energy shifts:
\begin{lem}
\label{lem:ground_state_energy_bounds}Let $|g|$ be sufficiently
small and $|P|<P_{max}$. For $1\leq n\leq N$ the following holds:
\begin{enumerate}[label=(\roman*)]
\item 
\begin{eqnarray}
E_{P,n}-E_{P,n-1} & = & -\Delta E_{P}|_{n}^{n-1}+{\cal O}\left(|g|^{4}\Lambda(1-\gamma)^{4/2}\right),\label{eq:energy_shift}\\
\Delta E_{P}|_{n}^{n-1} & := & g^{2}\left\langle \widehat{\Psi}_{P,n-1},\phi|_{n}^{n-1}\frac{1}{H_{P,n-1}-E_{P,n-1}}\phi^{*}|_{n}^{n-1}\widehat{\Psi}_{P,n-1}\right\rangle .\nonumber 
\end{eqnarray}

\item There exist universal constants $a,b>0$ such that 
\[
g^{2}a\Lambda\gamma^{n-1}(1-\gamma)\leq\Delta E_{P}|_{n}^{n-1}\leq g^{2}b\Lambda\gamma^{n-1}(1-\gamma).
\]

\end{enumerate}
\end{lem}
\begin{proof}
Claim (i) follows from the expansion formula of \thmref{expansion_formulas}
applied to
\[
E_{P,n}-E_{P,n-1}=\frac{\left\langle \Psi_{P,n},\left[H_{P,n}-H_{P,n-1}\right]\Psi_{P,n-1}\right\rangle }{\left\langle \Psi_{P,n},\Psi_{P,n-1}\right\rangle }=\frac{\left\langle \Psi_{P,n},g\Phi|_{n}^{n-1}\Psi_{P,n-1}\right\rangle }{\left\langle \Psi_{P,n},\Psi_{P,n-1}\right\rangle }.
\]

Next, we show claim (ii). The bound from above follows by using the
pull-through formula, i.e.,
\begin{equation}
\Delta E_{P}|_{n}^{n-1}=g^{2}\int_{\Lambda\gamma^{n}}^{\Lambda\gamma^{n-1}}dk\,\rho(k)^{2}\left\langle \widehat{\Psi}_{P,n-1},\frac{1}{H_{P-k,n-1}+\omega(k)-E_{P,n-1}}\widehat{\Psi}_{P,n-1}\right\rangle \label{eq:energy_shift_to_be_bound}
\end{equation}
and the estimate
\begin{equation}
g^{2}\int_{\Lambda\gamma^{n}}^{\Lambda\gamma^{n-1}}dk\,\rho(k)^{2}\left\langle \widehat{\Psi}_{P,n-1},\frac{1}{H_{P-k,n-1}+\omega(k)-E_{P,n-1}}\widehat{\Psi}_{P,n-1}\right\rangle \leq g^{2}b\Lambda\gamma^{n-1}(1-\gamma).\label{eq:delta_E_bound}
\end{equation}
that uses \lemref{energy} (iii). The bound from below of (\ref{eq:energy_shift_to_be_bound})
can be shown by a similar argument as in (iii) of \thmref{expansion_formulas}.
Therefore we omit the proof. 
\end{proof}
The established upper and lower bounds given in \lemref{ground_state_energy_bounds}
enable us to prove the first main result.
\begin{proof}[\textbf{Proof of \thmref{energy_bounds}}]
Using (i) of \lemref{ground_state_energy_bounds} we find
\[
E_{P,N}=E_{P,0}-\sum_{n=1}^{N}\Delta E_{P}|_{n}^{n-1}+{\cal O}\left(N|g|^{4}\Lambda(1-\gamma)^{4/2}\right)\,,
\]
where by construction $E_{P,0}=\sqrt{P^{2}+m^{2}}$.

The inequalities in (ii) of \lemref{ground_state_energy_bounds} imply
\begin{equation}
E_{P,N}\leq\sqrt{P^{2}+m^{2}}-g^{2}a\Lambda(1-\gamma)\sum_{n=1}^{N}\gamma^{n-1}+|g|^{4}C\Lambda N(1-\gamma)\label{eq:estimate-d-energy1}
\end{equation}
as well as
\begin{equation}
E_{P,N}\geq\sqrt{P^{2}+m^{2}}-g^{2}b\Lambda(1-\gamma)\sum_{n=1}^{N}\gamma^{n-1}-|g|^{4}C\Lambda\ln\Lambda(1-\gamma).\label{eq:estimate-d-energy2}
\end{equation}
Notice that by the same argument used in Lemma \ref{lem:gap_estimate}
one can conclude that $E_{P,N}=\inf\spec{H_{P,N}\restrict{\cal F}_{j}}$
for all $j\geq N$. Since $N=\frac{\ln\Lambda}{-\ln\gamma}$ and the
estimates in (\ref{eq:estimate-d-energy1}) and (\ref{eq:estimate-d-energy2})
hold for $\gamma$ arbitrarily close to $1$, they imply the inequalities
in (\ref{eq:ground_state_energy_estimate-1}).
\end{proof}

\section{\label{sec:eff_vel_eff_mass}The Effective Velocity and the Mass
Shell}

In this last section we provide the proof of \thmref{effective_velocity-2}, 
the starting point of which is the expression of the first derivative
of the ground state energies $E_{P,n}$ that follows from analytic
perturbation theory in $P$ as stated in the proposition below:
\begin{prop}
Suppose $E_{P,n}$ is the non-degenerate isolated eigenvalue corresponding
to the ground state $\Psi_{P,n}$. Then, the equation
\begin{eqnarray}
\frac{\partial E_{P,n}}{\partial P_{i}} & = & \left\langle \widehat{\Psi}_{P,n},V_{i}(P)\widehat{\Psi}_{P,n}\right\rangle \,,\,\quad\quad V_{i}(P):=\frac{P_{i}-P_{i}^{f}}{\left[(P-P^{f})^{2}+m^{2}\right]^{1/2}}\qquad\label{eq:energy_derivatives-1}
\end{eqnarray}
holds true for components $i=1,2,3$.\end{prop}
\begin{proof}
See Lemma 3.7 in \cite{fraehlich_existence_1974}.
\end{proof}
In order to control the scalar product in ($\ref{eq:energy_derivatives-1}$)
the following definition will be convenient:
\begin{defn}
For each $\Lambda\gamma^{n-1}$ we consider the energy level 
\begin{equation}
\min\left\{ \Lambda,\frac{\Lambda\gamma^{n-1}}{g^{\epsilon}}\right\} ,\qquad0<\epsilon\leq1/2\,\,,\label{eq:def_chi-1}
\end{equation}
and $l\in\mathbb{N}\cup\{0\}$ such that 
\[
\Lambda\gamma^{l}\leq\min\left\{ \Lambda,\frac{\Lambda\gamma^{n-1}}{g^{\epsilon}}\right\} <\Lambda\gamma^{l-1}\,.
\]
We define
\begin{equation}
\Xi{}_{n-1}:=\Lambda\gamma^{l}\,.\label{eq:def_chi2-1}
\end{equation}

\end{defn}
The energy scale $\Xi{}_{n-1}$ will be used in a convenient \emph{backwards
expansion} to gain a certain power of $|g|$ in some estimates. From
now on, we use the notation
\[
H_{P,\Xi{}_{n-1}}:=H_{P}|_{\Xi{}_{n-1}}^{\Lambda},\qquad\Psi_{P,\Xi{}_{n-1}}:=\Psi_{P,l}\,.
\]
The following lemma gives a justification for this type of expansion:
\begin{lem}
\label{lem:backwards_expansion-1}Let $|g|$ be sufficiently small,
$|P|<P_{max}$, and $0<\epsilon\leq1/2$. For $z\in\Gamma_{P,n-1}$
the bound
\begin{equation}
\left\Vert \left(\frac{1}{H_{P,\Xi{}_{n-1}}-z}\right)^{1/2}g\Phi|_{\Lambda\gamma^{n-1}}^{\Xi{}_{n-1}}\left(\frac{1}{H_{P,\Xi{}_{n-1}}-z}\right)^{1/2}\right\Vert _{{\cal F}_{n-1}}\leq|g|^{\delta}C,\qquad\delta:=1-\frac{\epsilon}{2},\label{eq:backwards_order_est-1-1}
\end{equation}
holds true. Consequently, the expansion formulas
\begin{equation}
\begin{split}\Psi_{P,n-1}^{(\Xi{}_{n-1})} & :={\cal Q}_{P,n-1}\Psi_{P,\Xi{}_{n-1}},\\
{\cal Q}_{P,n-1} & :=-\frac{1}{2\pi i}\oint_{\Gamma_{P,n-1}}\frac{dz}{H_{P,n-1}-z}\restrict{\cal F}_{n-1}\\
 & =-\frac{1}{2\pi i}\oint_{\Gamma_{P,n-1}}\frac{dz}{H_{P,\Xi{}_{n-1}}-z}\sum_{j=0}^{\infty}\left[-g\Phi|_{\Lambda\gamma^{n-1}}^{\Xi{}_{n-1}}\frac{1}{H_{P,\Xi{}_{n-1}}-z}\right]^{j}\restrict{\cal F}_{n-1}
\end{split}
\label{eq:norm relation-1}
\end{equation}
hold true and 
\begin{equation}
\|\Psi_{P,n-1}^{(\Xi{}_{n-1})}\|^{2}\geq(1-\mathcal{O}(|g|^{4\delta}))\|\Psi_{P,\Xi{}_{n-1}}\|^{2}\,.\label{eq:norm relation-2-1}
\end{equation}
\end{lem}
\begin{proof}
With the help of \lemref{ground_state_energy_bounds} we infer the
bound
\begin{equation}
\left|E_{P,n-1}-E_{P,\Xi{}_{n-1}}\right|\leq Cg^{2}\Xi{}_{n-1}.\label{eq:backward_energy_est-1}
\end{equation}
Hence, by the definition of $\Xi{}_{n-1}$ in (\ref{eq:def_chi-1})
and $0<\epsilon\leq1/2$, $|g|$ can be chosen sufficiently small
but uniformly in $n$ such that both ground state energies, $E_{P,n-1}$
and $E_{P,\Xi{}_{n-1}}$, lie inside the contour $\Gamma_{P,n-1}$.
We estimate
\begin{multline*}
\sup_{z\in\Gamma_{P,n-1}}\left\Vert \left(\frac{1}{H_{P,\Xi{}_{n-1}}-z}\right)^{1/2}g\Phi|_{\Lambda\gamma^{n-1}}^{\Xi{}_{n-1}}\left(\frac{1}{H_{P,\Xi{}_{n-1}}-z}\right)^{1/2}\right\Vert _{{\cal F}_{n-1}}\\
\leq2|g|\sup_{z\in\Gamma_{P,n-1}}\left\Vert \left(\frac{1}{H_{P,\Xi{}_{n-1}}-z}\right)^{1/2}\right\Vert _{{\cal F}_{n-1}}\cdot\sup_{z\in\Gamma_{P,n-1}}\left\Vert \phi|_{\Lambda\gamma^{n-1}}^{\Xi_{n-1}}\left(\frac{1}{H_{P,\Xi{}_{n-1}}-z}\right)^{1/2}\right\Vert _{{\cal F}_{n-1}}\,.
\end{multline*}
 A similar computation as in \lemref{gap_estimate} gives
\begin{equation}
\gap{H_{P,\Xi{}_{n-1}}}{{\cal F}_{n-1}}\geq\zeta\omega\left(\Lambda\gamma^{n}\right)\label{eq:reexpansion_gap-1}
\end{equation}
such that for sufficiently small $|g|$ one has the bound
\begin{equation}
\left\Vert \left(\frac{1}{H_{P,\Xi{}_{n-1}}-z}\right)^{1/2}\right\Vert _{{\cal F}_{n-1}}\leq\left(\frac{C}{\Lambda\gamma^{n}}\right)^{1/2}\label{eq:resolvent-wrong-contour}
\end{equation}
by using inequality (i) in Lemma \ref{lem:energy}. Furthermore, one
can bound

\begin{eqnarray*}
\left\Vert \phi|_{\Lambda\gamma^{n-1}}^{\Xi{}_{n-1}}\left(\frac{1}{H_{P,\Xi{}_{n-1}}-z}\right)^{1/2}\right\Vert _{{\cal F}_{n-1}} & \leq & C\Xi{}_{n-1}^{1/2}\left\Vert \left(H^{f}|_{\Lambda\gamma^{n-1}}^{\Xi{}_{n-1}}\right)^{1/2}\left(\frac{1}{H_{P,\Xi{}_{n-1}}-z}\right)^{1/2}\right\Vert _{{\cal F}_{n-1}}\leq C\Xi{}_{n-1}^{1/2}\theta^{-1/2}.
\end{eqnarray*}
Hence, we may conclude that
\[
\left\Vert \left(\frac{1}{H_{P,\Xi{}_{n-1}}-z}\right)^{1/2}g\Phi|_{\Lambda\gamma^{n-1}}^{\Xi{}_{n-1}}\left(\frac{1}{H_{P,\Xi{}_{n-1}}-z}\right)^{1/2}\right\Vert _{{\cal F}_{n-1}}\leq|g|C\begin{cases}
\left(\frac{\Lambda\gamma^{n-1}}{\Lambda\gamma^{n}g^{\epsilon}}\right)^{1/2} & \text{for }\frac{\Lambda\gamma^{n-1}}{g^{\epsilon}}<\Lambda\\
\left(\frac{\Lambda}{\Lambda\gamma^{n}}\right)^{1/2} & \text{for }\frac{\Lambda\gamma^{n-1}}{g^{\epsilon}}\geq\Lambda
\end{cases}\qquad\leq|g|^{\delta}C.
\]
This ensures the validity of the expansion formulas (\ref{eq:norm relation-1})
as well as the relation in (\ref{eq:norm relation-2-1}). 
\end{proof}
We can now prove our second main result:
\begin{proof}[\textbf{Proof of \thmref{effective_velocity-2}}]
The strategy of proof is an expansion using the formulas provided
by \thmref{expansion_formulas}. As a first observation we note that
by the spectral theorem the bounds 
\begin{equation}
\|V_{i}(P)\|\leq1\quad\forall P\in\mathbb{R}^{3},\qquad\left|\frac{\partial E_{P,n}}{\partial P_{i}}\right|\leq1\quad\text{for \quad}|P|<P_{max}\label{eq:vel_est}
\end{equation}
hold. These inequalities will be employed frequently without further
notice.

With
the help of \thmref{expansion_formulas} we find the following expansion
for all $N\geq n\geq1$: 
\begin{align}
\left\langle \widehat{\Psi}_{P,n},V_{i}(P)\widehat{\Psi}_{P,n}\right\rangle = & \frac{\left\langle \Psi_{P,n},V_{i}(P)\Psi_{P,n}\right\rangle }{\left\langle \Psi_{P,n},\Psi_{P,n}\right\rangle }\label{eq:T expansion}\\
= & \frac{1+g^{2}\alpha_{P}|_{n}^{n-1}+{\cal O}\left(|g|^{4}(1-\gamma)^{4/2}\right)}{\left\langle \Psi_{P,n-1},\Psi_{P,n-1}\right\rangle }\Bigg[\left\langle \Psi_{P,n-1},V_{i}(P)\Psi_{P,n-1}\right\rangle +\nonumber \\
 & +g^{2}\left\langle \frac{1}{H_{P,n-1}-E_{P,n-1}}\phi^{*}|_{n}^{n-1}\Psi_{P,n-1},V_{i}(P)\frac{1}{H_{P,n-1}-E_{P,n-1}}\phi^{*}|_{n}^{n-1}\Psi_{P,n-1}\right\rangle \nonumber \\
 & +g^{2}\left\langle \widetilde{\mathcal{Q}}_{P,n-1}^{\perp}\frac{1}{H_{P,n-1}-E_{P,n-1}}\phi|_{n}^{n-1}\frac{1}{H_{P,n-1}-E_{P,n-1}}\phi^{*}|_{n}^{n-1}\Psi_{P,n-1},V_{i}(P)\Psi_{P,n-1}\right\rangle +h.c.\nonumber \\
 & -g^{2}\left\langle \widetilde{\mathcal{Q}}_{P,n-1}\phi|_{n}^{n-1}\left(\frac{1}{H_{P,n-1}-E_{P,n-1}}\right)^{2}\phi^{*}|_{n}^{n-1}\Psi_{P,n-1},V_{i}(P)\Psi_{P,n-1}\right\rangle +h.c.\label{eq:term_3}\\
 & +{\cal O}\left(|g|^{4}(1-\gamma)^{4/2}\right)\Bigg].\nonumber 
\end{align}
 We observe that
\[
(\ref{eq:term_3})=-2g^{2}\alpha_{P}|_{n}^{n-1}\left\langle \Psi_{P,n-1},V_{i}(P)\Psi_{P,n-1}\right\rangle 
\]
because
\[
g^{2}\left\langle {\cal \widetilde{{\cal Q}}}_{P,n-1}\phi|_{n}^{n-1}\left(\frac{1}{H_{P,n-1}-E_{P,n-1}}\right)^{2}\phi^{*}|_{n}^{n-1}\Psi_{P,n-1},V_{i}(P)\Psi_{P,n-1}\right\rangle =g^{2}\alpha_{P}|_{n}^{n-1}\left\langle \Psi_{P,n-1},V_{i}(P)\Psi_{P,n-1}\right\rangle \,.
\]
Hence, we can rewrite (\ref{eq:T expansion}) as
\begin{align}
\left\langle \widehat{\Psi}_{P,n},V_{i}(P)\widehat{\Psi}_{P,n}\right\rangle = & \left(1-g^{2}\alpha_{P}|_{n}^{n-1}+{\cal O}\left(|g|^{4}(1-\gamma)^{4/2}\right)\right)\left\langle \widehat{\Psi}_{P,n-1},V_{i}(P)\widehat{\Psi}_{P,n-1}\right\rangle \label{eq:recursion}\\
 & +g^{2}\left\langle \frac{1}{H_{P,n-1}-E_{P,n-1}}\phi^{*}|_{n}^{n-1}\widehat{\Psi}_{P,n-1},V_{i}(P)\frac{1}{H_{P,n-1}-E_{P,n-1}}\phi^{*}|_{n}^{n-1}\widehat{\Psi}_{P,n-1}\right\rangle \label{eq:term_1}\\
 & +g^{2}2\Re\left\langle {\cal \widetilde{{\cal Q}}}_{P,n-1}^{\perp}\frac{1}{H_{P,n-1}-E_{P,n-1}}\phi|_{n}^{n-1}\frac{1}{H_{P,n-1}-E_{P,n-1}}\phi^{*}|_{n}^{n-1}\widehat{\Psi}_{P,n-1},V_{i}(P)\widehat{\Psi}_{P,n-1}\right\rangle \label{eq:term_2}\\
 & +{\cal O}\left(|g|^{4}(1-\gamma)^{4/2}\right).\nonumber 
\end{align}
Next, we proceed iteratively by expanding $\left\langle \widehat{\Psi}_{P,n},V_{i}(P)\widehat{\Psi}_{P,n}\right\rangle $
at each step from $n=N$ to $n=0$. Meanwhile, we define 
\[
A_{P,n-1}:=(\ref{eq:term_1}),\qquad B_{P,n-1}:=(\ref{eq:term_2})\,.
\]
As a result of the iteration we find the following expansion 
\begin{multline}
\left\langle \widehat{\Psi}_{P,N},V_{i}(P)\widehat{\Psi}_{P,N}\right\rangle =\prod_{j=1}^{N}\left(1-g^{2}\alpha_{P}|_{N-j+1}^{N-j}\right)\left\langle \widehat{\Psi}_{P,0},V_{i}(P)\widehat{\Psi}_{P,0}\right\rangle \\
+\sum_{j=2}^{N-1}\left(1-g^{2}\alpha_{P}|_{N}^{N-1}\right)\ldots\left(1-g^{2}\alpha_{P}|_{N-j+1}^{N-j}\right)\left[A_{P,N-j-1}+B_{P,N-j-1}\right]\\
+\left(1-g^{2}\alpha_{P}|_{N}^{N-1}\right)\left[A_{P,N-2}+B_{P,N-2}\right]+\left[A_{P,N-1}+B_{P,N-1}\right]+{\cal O}\left(|g|^{4}N(1-\gamma)^{4/2}\right).\label{eq:eff_mass_in_terms_of_0}
\end{multline}
Let us assume one could show the bounds 
\begin{eqnarray}
\left|A_{P,N-j}\right| & \leq & g^{2}C\frac{1-\gamma}{\Lambda\gamma^{N-j+1}},\label{eq:A_assumption}\\
\left|B_{P,N-j}\right| & \leq & |g|^{5/2}C(1-\gamma)\label{eq:B_assumption}
\end{eqnarray}
where we stress that the universal constant $C$ is independent of
the mass $m$. Then, using the following ingredients
\begin{itemize}
\item (iii) of \thmref{expansion_formulas},
\item $N=\frac{\ln\Lambda}{-\ln\gamma}$, 
\item the basic estimates
\[
\prod_{j=1}^{N}\left(1-g^{2}\alpha_{P}|_{N-j+1}^{N-j}\right)\leq\prod_{j=1}^{N}\left(1-g^{2}c_{1}(1-\gamma)\right)\leq\Lambda^{-g^{2}c_{1}\frac{(1-\gamma)}{-\ln\gamma}},
\]
\begin{multline*}
\sum_{j=2}^{N-1}\left(1-g^{2}\alpha_{P}|_{N}^{N-1}\right)\ldots\left(1-g^{2}\alpha_{P}|_{N-j+1}^{N-j}\right)+\left(1-g^{2}\alpha_{P}|_{N}^{N-1}\right)+1\leq\sum_{j=0}^{N-1}\left(1-g^{2}c_{1}(1-\gamma)\right)^{j}\\
\leq\frac{1}{g^{2}c_{1}(1-\gamma)}\,,
\end{multline*}
and using $\Lambda\gamma^{N}=1$ 
\begin{multline*}
\sum_{j=2}^{N-1}\left(1-g^{2}\alpha_{P}|_{N}^{N-1}\right)\ldots\left(1-g^{2}\alpha_{P}|_{N-j+1}^{N-j}\right)\frac{1-\gamma}{\Lambda\gamma^{N-j}}+\left(1-g^{2}\alpha_{P}|_{N}^{N-1}\right)\frac{1-\gamma}{\Lambda\gamma^{N-1}}+\frac{1-\gamma}{\Lambda\gamma^{N}}\\
\leq C(1-\gamma)\sum_{j=0}^{N-1}\gamma^{j}\leq C,
\end{multline*}

\end{itemize}
the bounds in (\ref{eq:A_assumption})-(\ref{eq:B_assumption}) are
seen to imply 
\begin{equation}
\left|\left\langle \widehat{\Psi}_{P,N},V_{i}(P)\widehat{\Psi}_{P,N}\right\rangle \right|\leq\Lambda^{-g^{2}c_{1}\frac{(1-\gamma)}{-\ln\gamma}}\frac{|P|}{\left[P^{2}+m^{2}\right]^{1/2}}+C|g|^{1/2}+C|g|^{4}\ln\Lambda(1-\gamma)\,,\label{eq:J_est}
\end{equation}
where we recall that $\left|\left\langle \widehat{\Psi}_{P,0},V_{i}(P)\widehat{\Psi}_{P,0}\right\rangle \right|=\frac{|P_{i}|}{\left[P^{2}+m^{2}\right]^{1/2}}$.

As the fineness parameter $\gamma$ can be chosen arbitrarily close
to one the bound in (\ref{thm:effective_velocity-1-1}) is proven.\\
We show now that the bounds (\ref{eq:A_assumption})-(\ref{eq:B_assumption})
hold true.\\
\\
\textbf{Bound} (\ref{eq:A_assumption})\textbf{:} Defining $P_{\lambda}:=\lambda P$
and its components $P_{\lambda\, i}:=\lambda P_{i}$, $1\leq i\leq3$,
we start with the identity
\begin{equation}
A_{P,n-1}=\int_{0}^{1}d\lambda\frac{d}{d\lambda}g^{2}\left\langle \frac{1}{H_{P_{\lambda},n-1}-E_{P_{\lambda},n-1}}\phi^{*}|_{n}^{n-1}\widehat{\Psi}_{P_{\lambda},n-1},V_{i}(P_{\lambda})\frac{1}{H_{P_{\lambda},n-1}-E_{P_{\lambda},n-1}}\phi^{*}|_{n}^{n-1}\widehat{\Psi}_{P_{\lambda},n-1}\right\rangle \label{eq:A_integral}
\end{equation}
that holds because of analytic perturbation theory in $P$ (see Lemma
3.7 in \cite{fraehlich_existence_1974}) and

\[
\left\langle \frac{1}{H_{0,n-1}-E_{0,n-1}}\phi^{*}|_{n}^{n-1}\widehat{\Psi}_{0,n-1},V_{i}(0)\frac{1}{H_{0,n-1}-E_{0,n-1}}\phi^{*}|_{n}^{n-1}\widehat{\Psi}_{0,n-1}\right\rangle =0
\]
by symmetry under rotational invariance of $H_{0,n-1}$, $E_{0,n-1}$
and $\widehat{\Psi}_{0,n-1}$. In order to estimate the integrand
\begin{multline}
g^{2}\frac{d}{d\lambda}\left\langle \frac{1}{H_{P_{\lambda},n-1}-E_{P_{\lambda},n-1}}\phi^{*}|_{n}^{n-1}\widehat{\Psi}_{P_{\lambda},n-1},V_{i}(P_{\lambda})\frac{1}{H_{P_{\lambda},n-1}-E_{P_{\lambda},n-1}}\phi^{*}|_{n}^{n-1}\widehat{\Psi}_{P_{\lambda},n-1}\right\rangle =\\
=\lim_{h\to0}\frac{g^{2}}{h}\Bigg[\left\langle \frac{1}{H_{P_{\lambda+h},n-1}-E_{P_{\lambda+h},n-1}}\phi^{*}|_{n}^{n-1}\widehat{\Psi}_{P_{\lambda+h},n-1},V_{i}(P_{\lambda+h})\frac{1}{H_{P_{\lambda+h},n-1}-E_{P_{\lambda+h},n-1}}\phi^{*}|_{n}^{n-1}\widehat{\Psi}_{P_{\lambda+h},n-1}\right\rangle \\
-\left\langle \frac{1}{H_{P_{\lambda},n-1}-E_{P_{\lambda},n-1}}\phi^{*}|_{n}^{n-1}\widehat{\Psi}_{P_{\lambda},n-1},V_{i}(P_{\lambda})\frac{1}{H_{P_{\lambda},n-1}-E_{P_{\lambda},n-1}}\phi^{*}|_{n}^{n-1}\widehat{\Psi}_{P_{\lambda},n-1}\right\rangle \Bigg]\label{eq:vertical_limit}
\end{multline}
we first observe that in expression (\ref{eq:vertical_limit}), at
least for small $|h|$, the vector $\widehat{\Psi}_{P_{\lambda+h},n-1}$
can be replaced by the vector $\widehat{\Upsilon}_{P_{\lambda+h},n-1}$
where 
\[
\Upsilon_{P_{\lambda+h},n-1}:=-\frac{1}{2\pi i}\oint_{\Gamma_{P,n-1}}\frac{dz}{H_{P_{\lambda+h,n-1}}-z}\Psi_{P_{\lambda},n-1}\,.
\]
Notice that $\Upsilon_{P_{\lambda+h},n-1}\parallel\Psi_{P_{\lambda+h},n-1}$
and $\Upsilon_{P_{\lambda+h},n-1}\big|_{h=0}=\Psi_{P_{\lambda},n-1}$.
Hence, we need to estimate three types of terms:
\begin{multline}
\lim_{h\to0}\frac{g^{2}}{h}\Bigg<\left[\frac{1}{H_{P_{\lambda+h},n-1}-E_{P_{\lambda+h},n-1}}-\frac{1}{H_{P_{\lambda},n-1}-E_{P_{\lambda},n-1}}\right]\phi^{*}|_{n}^{n-1}\widehat{\Psi}_{P_{\lambda},n-1},\\
V_{i}(P_{\lambda})\frac{1}{H_{P_{\lambda},n-1}-E_{P_{\lambda},n-1}}\phi^{*}|_{n}^{n-1}\widehat{\Psi}_{P_{\lambda},n-1}\Bigg>,\label{eq:type_d_res}
\end{multline}
\begin{equation}
\lim_{h\to0}\frac{g^{2}}{h}\left\langle \frac{1}{H_{P_{\lambda},n-1}-E_{P_{\lambda},n-1}}\phi^{*}|_{n}^{n-1}\left[\widehat{\Upsilon}_{P_{\lambda+h},n-1}-\widehat{\Upsilon}_{P_{\lambda},n-1}\right],V_{i}(P_{\lambda})\frac{1}{H_{P_{\lambda},n-1}-E_{P_{\lambda},n-1}}\phi^{*}|_{n}^{n-1}\widehat{\Upsilon}_{P_{\lambda},n-1}\right\rangle ,\label{eq:type_d_psi}
\end{equation}
\begin{equation}
\lim_{h\to0}\frac{g^{2}}{h}\left\langle \frac{1}{H_{P_{\lambda},n-1}-E_{P_{\lambda},n-1}}\phi^{*}|_{n}^{n-1}\widehat{\Psi}_{P_{\lambda},n-1},\left[V_{i}(P_{\lambda+h})-V_{i}(P_{\lambda})\right]\frac{1}{H_{P_{\lambda},n-1}-E_{P_{\lambda},n-1}}\phi^{*}|_{n}^{n-1}\widehat{\Psi}_{P_{\lambda},n-1}\right\rangle ,\label{eq:type_d_vel}
\end{equation}
In order to estimate term (\ref{eq:type_d_res}) we observe that the
expression is well defined because the vector $\phi^{*}|_{n}^{n-1}\widehat{\Psi}_{P_{\lambda},n-1}$
is orthogonal to the ground state vector of both the Hamiltonians
$H_{P_{\lambda+h},n-1}$ and $H_{P_{\lambda},n-1}$. Hence, we verify
that\foreignlanguage{english}{
\begin{multline}
\lim_{h\to0}\frac{1}{h}\left[\frac{1}{H_{P_{\lambda+h},n-1}-E_{P_{\lambda+h},n-1}}-\frac{1}{H_{P_{\lambda},n-1}-E_{P_{\lambda},n-1}}\right]\\
=\lim_{h\to0}\frac{1}{h}\left[\frac{1}{H_{P_{\lambda+h},n-1}-E_{P_{\lambda+h},n-1}}\left(H_{P_{\lambda},n-1}-H_{P_{\lambda+h},n-1}-E_{P_{\lambda},n-1}+E_{P_{\lambda+h},n-1}\right)\frac{1}{H_{P_{\lambda},n-1}-E_{P_{\lambda},n-1}}\right]\\
=\frac{1}{H_{P_{\lambda},n-1}-E_{P_{\lambda},n-1}}\left(-\frac{d}{d\lambda}\sqrt{(P_{\lambda}-P^{f})^{2}+m^{2}}+\frac{d}{d\lambda}E_{P_{\lambda},n-1}\right)\frac{1}{H_{P_{\lambda},n-1}-E_{P_{\lambda},n-1}}\\
=\frac{1}{H_{P_{\lambda},n-1}-E_{P_{\lambda},n-1}}\sum_{i=1}^{3}P_{\lambda\, i}\left(-V_{i}(P_{\lambda})+\frac{\partial E_{P,n-1}}{\partial P_{i}}\bigg|_{P\equiv P_{\lambda}}\right)\frac{1}{H_{P_{\lambda},n-1}-E_{P_{\lambda},n-1}}\label{eq:vertical expansion term 1)}
\end{multline}
}holds true when applied to the vector $\phi^{*}|_{n}^{n-1}\widehat{\Psi}_{P_{\lambda},n-1}$.
At first we treat the term proportional to $\sum_{i=1}^{3}P_{\lambda\, i}V_{i}(P_{\lambda})$.
Using (iii) in \lemref{energy}, the estimate in (\ref{eq:vel_est}),
and the pull-through formula, we get the estimate
\begin{multline*}
g^{2}\Bigg|\Bigg<V_{i}(P_{\lambda})\frac{1}{H_{P_{\lambda},n-1}-E_{P_{\lambda},n-1}}\phi^{*}|_{n}^{n-1}\widehat{\Psi}_{P_{\lambda},n-1},\\
\frac{1}{H_{P_{\lambda},n-1}-E_{P_{\lambda},n-1}}\sum_{j=1}^{3}P_{\lambda\, j}V_{j}(P_{\lambda})\frac{1}{H_{P_{\lambda},n-1}-E_{P_{\lambda},n-1}}\phi^{*}|_{n}^{n-1}\widehat{\Psi}_{P_{\lambda},n-1}\Bigg>\Bigg|\\
=g^{2}\Bigg|\int_{\Lambda\gamma^{n}}^{\Lambda\gamma^{n-1}}dk\rho(k)^{2}\Bigg<V_{i}(P_{\lambda}-k)\left(\frac{1}{H_{P_{\lambda}-k,n-1}-E_{P_{\lambda},n-1}+\omega(k)}\right)\widehat{\Psi}_{P_{\lambda},n-1},\\
\left(\frac{1}{H_{P_{\lambda}-k,n-1}-E_{P_{\lambda},n-1}+\omega(k)}\right)\sum_{j=1}^{3}P_{\lambda\, j}V_{j}(P_{\lambda}-k)\frac{1}{H_{P_{\lambda}-k,n-1}-E_{P_{\lambda},n-1}+\omega(k)}\widehat{\Psi}_{P_{\lambda},n-1}\Bigg>\Bigg|\\
\leq g^{2}C\int_{\Lambda\gamma^{n}}^{\Lambda\gamma^{n-1}}dk\,\frac{1}{|k|^{4}}\leq g^{2}\frac{C(1-\gamma)}{\Lambda\gamma^{n}}.
\end{multline*}
The remaining term in (\ref{eq:type_d_res}) being proportional to
$\sum_{i=1}^{3}P_{\lambda\, i}\left(\frac{\partial E_{P,n-1}}{\partial P_{i}}\big|_{P\equiv P_{\lambda}}\right)$
can be estimated in the same way. In consequence, we get 
\begin{equation}
\left|(\ref{eq:type_d_res})\right|\leq g^{2}\frac{C(1-\gamma)}{\Lambda\gamma^{n}}.\label{eq:type_d_res_final}
\end{equation}

Next, we consider term (\ref{eq:type_d_psi}). Using the differentiability
in $\lambda$ again we find
\begin{multline}
\lim_{h\to0}\frac{\Upsilon_{P_{\lambda+h},n-1}-\Upsilon_{P_{\lambda},n-1}}{h}=-\frac{1}{2\pi i}\lim_{h\to0}\frac{1}{h}\oint_{\Gamma_{P,n-1}}dz\,\left[\frac{1}{H_{P_{\lambda+h},n-1}-z}-\frac{1}{H_{P_{\lambda},n-1}-z}\right]\Psi_{P_{\lambda},n-1}\\
=-\frac{1}{2\pi i}\lim_{h\to0}\frac{1}{h}\oint_{\Gamma_{P,n-1}}dz\,\left[\frac{1}{H_{P_{\lambda},n-1}-z}\left(H_{P_{\lambda,n-1}}-H_{P_{\lambda+h,n-1}}\right)\frac{1}{H_{P_{\lambda},n-1}-z}\right]\Psi_{P_{\lambda},n-1}\\
=-\frac{1}{2\pi i}\oint_{\Gamma_{P,n-1}}dz\,\left[\frac{1}{H_{P_{\lambda},n-1}-z}\left(-\sum_{i=1}^{3}P_{\lambda\, i}V_{i}(P_{\lambda})\right)\frac{1}{H_{P_{\lambda},n-1}-z}\right]\Psi_{P_{\lambda},n-1}\\
=-{\cal Q}_{P_{\lambda},n-1}^{\perp}\frac{1}{H_{P_{\lambda},n-1}-E_{P_{\lambda},n-1}}\sum_{i=1}^{3}P_{\lambda\, i}V_{i}(P_{\lambda})\Psi_{P_{\lambda},n-1}\label{eq:d_ups-1}
\end{multline}
and
\begin{equation}
\lim_{h\to0}\frac{1}{h}\left[\frac{1}{\left\Vert \Upsilon_{P_{\lambda+h},n-1}\right\Vert }-\frac{1}{\left\Vert \Upsilon_{P_{\lambda},n-1}\right\Vert }\right]=-\frac{1}{\left\Vert \Upsilon_{P_{\lambda},n-1}\right\Vert ^{3}}\lim_{h\to0}\Re\left\langle \frac{\Upsilon_{P_{\lambda+h},n-1}-\Upsilon_{P_{\lambda},n-1}}{h},\Upsilon_{P_{\lambda},n-1}\right\rangle =0.\label{eq:d_abs_ups}
\end{equation}
Equations (\ref{eq:d_ups-1}) and (\ref{eq:d_abs_ups}), the pull-through
formula, and the gap estimate in \thmref{mass_shell} give
\begin{equation}
\left|(\ref{eq:type_d_psi})\right|\leq g^{2}\frac{C(1-\gamma)}{\Lambda\gamma^{n}}.\label{eq:type_d_psi_final}
\end{equation}

In the estimate of the third term, i.e., term (\ref{eq:type_d_vel}),
we exploit the additional decay which we gain through the derivative
of $V_{i}(P_{\lambda})$, i.e.,
\[
\lim_{h\to0}\frac{1}{h}\left[V_{i}(P_{\lambda+h})-V_{i}(P_{\lambda})\right]=\frac{P_{\lambda\, i}-V_{i}(P_{\lambda})\sum_{j=1}^{3}V_{j}(P_{\lambda})P_{\lambda\, j}}{\sqrt{(P_{\lambda}-P^{f})^{2}+m^{2}}}.
\]
Thus, we can rewrite and estimate (\ref{eq:type_d_vel}) as follows

\begin{multline}
\Bigg|g^{2}\int_{\Lambda\gamma^{n}}^{\Lambda\gamma^{n-1}}dk\,\rho(k)^{2}\Bigg<\widehat{\Psi}_{P_{\lambda},n-1},\frac{1}{H_{P_{\lambda}-k,n-1}+\omega(k)-E_{P_{\lambda},n-1}}\times\\
\left[\frac{P_{\lambda\, i}-V_{i}(P_{\lambda}-k)\sum_{j=1}^{3}V_{j}(P_{\lambda}-k)P_{\lambda\, j}}{\sqrt{(P_{\lambda}-P^{f}-k)^{2}+m^{2}}}\right]\frac{1}{H_{P_{\lambda}-k,n-1}+\omega(k)-E_{P_{\lambda},n-1}}\widehat{\Psi}_{P_{\lambda},n-1}\Bigg>\Bigg|\\
\leq Cg^{2}P_{max}\Bigg|\int_{\Lambda\gamma^{n}}^{\Lambda\gamma^{n-1}}dk\,\rho(k)^{2}\Bigg<\widehat{\Psi}_{P,n-1},\frac{1}{H_{P_{\lambda}-k,n-1}+\omega(k)-E_{P_{\lambda},n-1}}\times\\
\times\frac{1}{\sqrt{(P_{\lambda}-P^{f}-k)^{2}+m^{2}}}\,\frac{1}{H_{P_{\lambda}-k,n-1}+\omega(k)-E_{P_{\lambda},n-1}}\widehat{\Psi}_{P_{\lambda},n-1}\Bigg>\Bigg|\label{eq:apply_spec}
\end{multline}
where we have used the pull-through formula. Next we consider the
spectral measure $d\mu_{k}(\xi)\equiv f_{k}(\xi)d\xi$ (where $f_{k}(\xi)\geq0$
a.e.) associated with the vector 
\[
\frac{1}{H_{P_{\lambda}-k,n-1}+\omega(k)-E_{P_{\lambda},n-1}}\widehat{\Psi}_{P_{\lambda},n-1}
\]
in the joint spectral representation of the components of the operator
$P^{f}$ where $\xi$ is the spectral variable. The measure is defined
by 
\[
(0\leq)\,\,\|\chi_{\Omega}\frac{1}{H_{P_{\lambda}-k,n-1}+\omega(k)-E_{P_{\lambda},n-1}}\widehat{\Psi}_{P_{\lambda},n-1}\|^{2}=:\int_{\sigma(P^{f})}d\xi f_{k}(\xi)\,\chi_{\Omega}(\xi)\,\,\leq\frac{C}{|k|^{2}}
\]
for every measurable set $\Omega\subseteq\sigma(P^{f})$ where $\chi_{\Omega}(\xi)$
is the characteristic function of the set $\Omega$ and $\chi_{\Omega}$
is the corresponding spectral projection. Thus we can write (\ref{eq:apply_spec})
as follows
\begin{eqnarray}
(\ref{eq:apply_spec}) & = & Cg^{2}\int_{\Lambda\gamma^{n}}^{\Lambda\gamma^{n-1}}dk\,\frac{1}{|k|}\Big|\Big|\Big[\frac{1}{\sqrt{(P_{\lambda}-P^{f}-k)^{2}+m^{2}}}\Big]^{\frac{1}{2}}\frac{1}{H_{P_{\lambda}-k,n-1}+\omega(k)-E_{P_{\lambda},n-1}}\widehat{\Psi}_{P_{\lambda},n-1}\Big|\Big|^{2}\nonumber \\
 & = & Cg^{2}\int_{\Lambda\gamma^{n}}^{\Lambda\gamma^{n-1}}\int d\Omega_{k}d|k|\frac{1}{|k|}\int_{\sigma(P^{f})}d\xi\: f_{k}(\xi)\,\frac{1}{\sqrt{(P_{\lambda}-\xi-k)^{2}+m^{2}}}\label{eq:theta int}
\end{eqnarray}
By knowing that 
\[
\int_{\sigma(P^{f})}d\xi\, f_{k}(\xi)\,\frac{1}{\sqrt{(P_{\lambda}-\xi-k)^{2}+m^{2}}}<+\infty
\]
we can interchange the integration in $d\xi$ with the angular integration
in the variable $k$, i.e., 
\[
(\ref{eq:theta int})=Cg^{2}\int_{\Lambda\gamma^{n}}^{\Lambda\gamma^{n-1}}|k|d|k|\,\int_{\sigma(P^{f})}d\xi\int_{0}^{2\pi}d\varphi\int_{0}^{\pi}d\theta\sin\theta\: f_{k}(\xi)\,\frac{1}{\sqrt{(P_{\lambda}-\xi)^{2}+k^{2}-2\cos\theta|P_{\lambda}-\xi|\,|k|+m^{2}}}
\]
where $\theta$ denotes the angle between the vector $k$ and the
vector $P_{\lambda}-\xi$ and $\varphi$ the azimuthal angle with
respect to an arbitrarily chosen vector orthogonal to $P_{\lambda}-\xi$.
We split the integration in the variable $\theta$ into two regions:
$\theta\in\left[\frac{\pi}{3},\pi\right]$ and $\theta\in\left[0,\frac{\pi}{3}\right]$.
For $\theta\in\left[\frac{\pi}{3},\pi\right]$ being $\cos\theta\in\left[-1,\frac{1}{2}\right]$
we observe that 
\[
(P_{\lambda}-\xi)^{2}+k^{2}-2\cos\theta|P_{\lambda}-\xi|\,|k|\geq(P_{\lambda}-\xi)^{2}+k^{2}-|P_{\lambda}-\xi|\,|k|\geq\frac{3}{4}k^{2}
\]
and, consequently, 
\begin{align}
 & \int_{\sigma(P^{f})}d\xi\int_{0}^{2\pi}d\varphi\int_{\pi/3}^{\pi}d\theta\,\sin\theta\, f_{k}(\xi)\,\frac{1}{\sqrt{(P_{\lambda}-\xi)^{2}+k^{2}-2\cos\theta|P_{\lambda}-\xi|\,|k|+m^{2}}}\label{eq:theta_region_1}\\
\leq & C\int_{\sigma(P^{f})}d\xi\int_{0}^{2\pi}d\varphi\int_{0}^{\pi}d\theta\,\sin\theta\, f_{k}(\xi)\,\frac{1}{|k|}\\
\leq & \frac{C}{|k|}\int d\Omega_{k}\,\int_{\sigma(P^{f})}d\mu_{k}(\xi)\\
\leq & \frac{C}{|k|^{3}}\label{eq:constant-67}
\end{align}
Notice that the constant $C$ in (\ref{eq:constant-67}) can be chosen
to be independent of the mass $m$. Next, we treat the integration
over $\theta\in\left[0,\frac{\pi}{3}\right]$ where $\cos\theta\in\left[\frac{1}{2},1\right]$
and 
\[
(P_{\lambda}-\xi)^{2}+k^{2}-2\cos\theta|P_{\lambda}-\xi|\,|k|\geq\left[(P_{\lambda}-\xi)^{2}+k^{2}\right](1-\cos\theta)
\]
we find 
\begin{equation}
\begin{split}\int_{\sigma(P^{f})}d\xi\int_{0}^{2\pi}d\varphi\int_{0}^{\pi/3}d\theta\,\sin\theta\, f_{k}(\xi)\,\frac{1}{\sqrt{[(P_{\lambda}-\xi)^{2}+k^{2}](1-\cos\theta)+m^{2}}}\\
\leq\int_{\sigma(P^{f})}d\xi\,\int_{0}^{2\pi}d\varphi\int_{0}^{\pi/3}d\theta\,\sin\theta\frac{1}{|k|}\, f_{k}(\xi)\,\frac{1}{\sqrt{(1-\cos\theta)}}\\
\leq C\int_{\sigma(P^{f})}d\xi\,\int_{0}^{2\pi}d\varphi\int_{0}^{\pi/3}d\theta\frac{1}{|k|}\, f_{k}(\xi)\\
\leq\frac{C}{|k|^{3}}
\end{split}
\label{eq:theta_region_2}
\end{equation}
Notice that also the constant $C$ in (\ref{eq:theta_region_2}) can
be chosen to be independent of the mass $m$. Combining the results
for the two integration domains, i.e., (\ref{eq:theta_region_1})
and (\ref{eq:theta_region_2}), we arrive at
\begin{equation}
(\ref{eq:theta int})\leq g^{2}C\int_{\Lambda\gamma^{n}}^{\Lambda\gamma^{n-1}}d|k|\frac{1}{|k|^{2}}\leq g^{2}C\frac{1-\gamma}{\Lambda\gamma^{n}}.\label{eq:type_d_vel_final}
\end{equation}
Hence, we have proven the bound in (\ref{eq:type_d_vel}).

With the three bounds in (\ref{eq:type_d_res_final}), (\ref{eq:type_d_psi_final})
and (\ref{eq:type_d_vel_final}) we can control the integrand (\ref{eq:type_d_res})-(\ref{eq:type_d_vel}),
and hence, the integral given in (\ref{eq:A_integral}) which proves
the bound in (\ref{eq:A_assumption}).\\
\\
\textbf{Bound} (\ref{eq:B_assumption})\textbf{:} As a next step we
proceed with the bound of (\ref{eq:B_assumption}) where by using
the pull-through formula we get
\begin{multline}
\left|B_{P,n-1}\right|\\
=g^{2}\left|2\Re\int_{\Lambda\gamma^{n}}^{\Lambda\gamma^{n-1}}dk\,\rho(k)^{2}\left\langle {\cal Q}_{P,n-1}^{\perp}\frac{1}{H_{P,n-1}-E_{P,n-1}}\frac{1}{H_{P-k,n-1}+\omega(k)-E_{P,n-1}}\widehat{\Psi}_{P,n-1},V_{i}(P)\widehat{\Psi}_{P,n-1}\right\rangle \right|\\
\leq g^{2}C\int_{\Lambda\gamma^{n}}^{\Lambda\gamma^{n-1}}dk\,\frac{1}{k^{2}}\left\Vert \frac{1}{H_{P,n-1}-E_{P,n-1}}{\cal Q}_{P,n-1}^{\perp}V_{i}(P)\widehat{\Psi}_{P,n-1}\right\Vert \\
\leq g^{2}C\Lambda\gamma^{n-1}(1-\gamma)\left\Vert \frac{1}{H_{P,n-1}-E_{P,n-1}}{\cal Q}_{P,n-1}^{\perp}V_{i}(P)\widehat{\Psi}_{P,n-1}\right\Vert .\label{eq:B term 1}
\end{multline}
We shall now show that
\begin{equation}
\left\Vert \frac{1}{H_{P,n-1}-E_{P,n-1}}{\cal Q}_{P,n-1}^{\perp}V_{i}(P)\widehat{\Psi}_{P,n-1}\right\Vert \leq C\frac{|g|^{1/2}}{\Lambda\gamma^{n}}\label{eq:bad_term_est}
\end{equation}
holds true, so that, by inserting this bound in (\ref{eq:B term 1}),
we get the desired $m$-independent estimate in (\ref{eq:B_assumption}). 

In order to gain a certain power of $|g|$ we re-expand the left-hand
side of (\ref{eq:bad_term_est}) backwards from energy level $\Lambda\gamma^{n-1}$
to $\Xi{}_{n-1}$, as defined in (\ref{eq:def_chi2-1}), with the
help of \lemref{backwards_expansion-1} for an $\epsilon$, $0<\epsilon\leq\frac{1}{2}$,
and $\delta=1-\frac{\epsilon}{2}$ which will be fixed later. We know
that
\begin{itemize}
\item $\left\Vert \left(\frac{1}{H_{P,\Xi{}_{n-1}}-z}\right)^{1/2}g\Phi|_{\Lambda\gamma^{n-1}}^{\Xi{}_{n-1}}\left(\frac{1}{H_{P,\Xi{}_{n-1}}-z}\right)^{1/2}\right\Vert _{{\cal F}_{n-1}}\leq|g|^{\delta}C\quad$
for $z\in\Gamma_{P,n-1}$ (see (\ref{eq:backwards_order_est-1-1})),
\item $\Psi_{P,n-1}^{(\Xi{}_{n-1})}$ and $\Psi_{P,n-1}$ are two vectors
belonging to the same ray with $\|\Psi_{P,n-1}^{(\Xi{}_{n-1})}\|^{2}\geq(1-\mathcal{O}(|g|^{4\delta}))\|\Psi_{P,\Xi{}_{n-1}}\|^{2}$
(see (\ref{eq:norm relation-1})). 
\end{itemize}
Thus, denoting the length of the contour $\Gamma_{P,n-1}$ by $|\Gamma_{P,n-1}|$,
we find for $|g|$ sufficiently small 
\begin{multline}
\left\Vert \frac{1}{H_{P,n-1}-E_{P,n-1}}{\cal Q}_{P,n-1}^{\perp}V_{i}(P)\widehat{\Psi}_{P,n-1}\right\Vert =\left\Vert {\cal Q}_{P,n-1}^{\perp}\frac{1}{H_{P,n-1}-E_{P,n-1}}{\cal Q}_{P,n-1}^{\perp}V_{i}(P)\widehat{\Psi}_{P,n-1}^{(\Xi{}_{n-1})}\right\Vert \\
\leq C\left\Vert {\cal Q}_{P,n-1}^{\perp}\frac{1}{H_{P,n-1}-E_{P,n-1}}{\cal Q}_{P,n-1}^{\perp}V_{i}(P)\widehat{\Psi}_{P,\Xi{}_{n-1}}\right\Vert \\
+C\left\Vert \frac{1}{H_{P,n-1}-E_{P,n-1}}{\cal Q}_{P,n-1}^{\perp}\right\Vert _{{\cal F}_{n-1}}\left\Vert V_{i}(P)\right\Vert |\Gamma_{P,n-1}|\sup_{z\in\Gamma_{P,n-1}}\sum_{j=1}^{\infty}\left\Vert \frac{1}{H_{P,\Xi{}_{n-1}}-z}\left[-g\phi^{*}|_{\Lambda\gamma^{n-1}}^{\Xi{}_{n-1}}\frac{1}{H_{P,\Xi{}_{n-1}}-z}\right]^{j}\widehat{\Psi}_{P,\Xi{}_{n-1}}\right\Vert \\
\leq\left\Vert {\cal Q}_{P,n-1}^{\perp}\frac{1}{H_{P,n-1}-E_{P,n-1}}{\cal Q}_{P,n-1}^{\perp}V_{i}(P)\widehat{\Psi}_{P,\Xi{}_{n-1}}\right\Vert +C\frac{|g|^{\delta}}{\Lambda\gamma^{n}}\label{eq:backward_exp_step1}
\end{multline}
where we have used the bound in (\ref{eq:backwards_order_est-1-1}),
the inequality in (\ref{eq:resolvent-wrong-contour}), and the gap
estimate given in \thmref{mass_shell}. Using the same ingredients,
we estimate
\begin{equation}
\left\Vert {\cal Q}_{P,n-1}^{\perp}\frac{1}{H_{P,n-1}-E_{P,n-1}}{\cal Q}_{P,n-1}^{\perp}V_{i}(P)\widehat{\Psi}_{P,\Xi{}_{n-1}}\right\Vert \label{eq:backwards expansion step 1.5}
\end{equation}
 by expanding the spectral projection on the right, i.e.,
\begin{multline}
(\ref{eq:backwards expansion step 1.5})\leq\left\Vert {\cal Q}_{P,n-1}^{\perp}\frac{1}{H_{P,n-1}-E_{P,n-1}}{\cal Q}_{P,\Xi{}_{n-1}}^{\perp}V_{i}(P)\widehat{\Psi}_{P,\Xi{}_{n-1}}\right\Vert \\
+C\left\Vert {\cal Q}_{P,n-1}^{\perp}\frac{1}{H_{P,n-1}-E_{P,n-1}}\right\Vert _{{\cal F}_{n-1}}|\Gamma_{P,n-1}|\sup_{z\in\Gamma_{P,n-1}}\sum_{j=1}^{\infty}\left\Vert \frac{1}{H_{P,\Xi{}_{n-1}}-z}\left[-g\Phi|_{\Lambda\gamma^{n-1}}^{\Xi{}_{n-1}}\frac{1}{H_{P,\Xi{}_{n-1}}-z}\right]^{j}\right\Vert _{{\cal F}_{n-1}}\left\Vert V_{i}(P)\right\Vert \\
\leq\left\Vert {\cal Q}_{P,n-1}^{\perp}\frac{1}{H_{P,n-1}-E_{P,n-1}}{\cal Q}_{P,\Xi{}_{n-1}}^{\perp}V_{i}(P)\widehat{\Psi}_{P,\Xi{}_{n-1}}\right\Vert +C\frac{|g|^{\delta}}{\Lambda\gamma^{n}}\label{eq:backward_exp_step2}
\end{multline}
where ${\cal Q}_{P,\Xi{}_{n-1}}^{\perp}\equiv{\cal Q}_{P,l}^{\perp}$
, for some $l\leq N$ specified in (\ref{eq:def_chi2-1}). Next, we
study 
\begin{equation}
\left\Vert {\cal Q}_{P,n-1}^{\perp}\frac{1}{H_{P,n-1}-E_{P,n-1}}{\cal Q}_{P,\Xi{}_{n-1}}^{\perp}V_{i}(P)\widehat{\Psi}_{P,\Xi{}_{n-1}}\right\Vert \label{eq:backwards expansion step 2.5}
\end{equation}
by applying the resolvent formula
\begin{multline}
(\ref{eq:backwards expansion step 2.5})\leq\left\Vert {\cal Q}_{P,n-1}^{\perp}\frac{1}{H_{P,\Xi{}_{n-1}}-E_{P,n-1}}{\cal Q}_{P,\Xi{}_{n-1}}^{\perp}V_{i}(P)\widehat{\Psi}_{P,\Xi{}_{n-1}}\right\Vert \\
+\left\Vert {\cal Q}_{P,n-1}^{\perp}\frac{1}{H_{P,n-1}-E_{P,n-1}}g\phi^{*}|_{\Lambda\gamma^{n-1}}^{\Xi{}_{n-1}}\frac{1}{H_{P,\Xi{}_{n-1}}-E_{P,n-1}}{\cal Q}_{P,\Xi{}_{n-1}}^{\perp}V_{i}(P)\widehat{\Psi}_{P,\Xi{}_{n-1}}\right\Vert .\label{eq:backward_exp_step3}
\end{multline}
In order to estimate (\ref{eq:backward_exp_step3}) we make use of
the following intermediate steps:
\begin{itemize}
\item $\left\Vert {\cal Q}_{P,n-1}^{\perp}\left(\frac{1}{H_{P,n-1}-E_{P,n-1}}\right)\right\Vert _{{\cal F}_{n-1}}\leq\frac{C}{\Lambda\gamma^{n}}$,
\item 
\begin{multline*}
\left\Vert g\phi^{*}|_{\Lambda\gamma^{n-1}}^{\Xi{}_{n-1}}\frac{1}{H_{P,\Xi{}_{n-1}}-E_{P,n-1}}{\cal Q}_{P,\Xi{}_{n-1}}^{\perp}V_{i}(P)\widehat{\Psi}_{P,\Xi{}_{n-1}}\right\Vert \\
=|g|\left(\int_{\Lambda\gamma^{n-1}}^{\Xi{}_{n-1}}dk\,\rho(k)^{2}\right)^{1/2}\left\Vert \frac{1}{H_{P,\Xi{}_{n-1}}-E_{P,n-1}}{\cal Q}_{P,\Xi{}_{n-1}}^{\perp}V_{i}(P)\widehat{\Psi}_{P,\Xi{}_{n-1}}\right\Vert \\
\leq C|g|\Xi{}_{n-1}\frac{1}{\Xi{}_{n-1}}
\end{multline*}
following from\\
\begin{equation}
\left\Vert {\cal Q}_{P,\Xi{}_{n-1}}^{\perp}\frac{1}{H_{P,\Xi{}_{n-1}}-E_{P,n-1}}\right\Vert _{{\cal F}|_{\Xi{}_{n-1}}^{\Lambda}}\leq\frac{C}{\Xi{}_{n-1}}=C\max\Big(\frac{g^{\epsilon}}{\Lambda\gamma^{n}};\frac{1}{\Lambda}\Big)\label{eq:backward-gap-estimate}
\end{equation}
that holds because of \thmref{mass_shell} and inequality (i) in Lemma
\ref{lem:energy}.
\end{itemize}
This implies
\begin{equation}
(\ref{eq:backwards expansion step 2.5})\leq\left\Vert {\cal Q}_{P,n-1}^{\perp}\frac{1}{H_{P,\Xi{}_{n-1}}-E_{P,n-1}}{\cal Q}_{P,\Xi{}_{n-1}}^{\perp}V_{i}(P)\widehat{\Psi}_{P,\Xi{}_{n-1}}\right\Vert +C\frac{|g|}{\Lambda\gamma^{n}}.\label{eq:backward_exp_step4}
\end{equation}
Next we consider
\begin{equation}
\left\Vert {\cal Q}_{P,n-1}^{\perp}\frac{1}{H_{P,\Xi{}_{n-1}}-E_{P,n-1}}{\cal Q}_{P,\Xi{}_{n-1}}^{\perp}V_{i}(P)\widehat{\Psi}_{P,\Xi{}_{n-1}}\right\Vert \label{eq:backward expansion step 4.5}
\end{equation}
 and re-expand the first spectral projection. Hence, by using (\ref{eq:backwards_order_est-1-1})
and (\ref{eq:backward-gap-estimate}) we can conclude that
\begin{equation}
(\ref{eq:backward expansion step 4.5})\leq\left\Vert {\cal Q}_{P,\Xi{}_{n-1}}^{\perp}\frac{1}{H_{P,\Xi_{n-1}}-E_{P,n-1}}{\cal Q}_{P,\Xi{}_{n-1}}^{\perp}V_{i}(P)\widehat{\Psi}_{P,\Xi{}_{n-1}}\right\Vert +C\frac{|g|^{\delta}}{\Lambda\gamma^{n}}.\label{eq:backward_exp_step5}
\end{equation}
As a last step, for the first term on the right-hand side of (\ref{eq:backward_exp_step5})
we have to regard two cases:
\begin{enumerate}
\item Case $\Xi{}_{n-1}<\Lambda.$ In this case we exploit
\[
\left\Vert {\cal Q}_{P,\Xi{}_{n-1}}^{\perp}\frac{1}{H_{P,\Xi{}_{n-1}}-E_{P,n-1}}\right\Vert _{{\cal F}|_{\Xi{}_{n-1}}^{\Lambda}}\leq\frac{g^{\epsilon}C}{\Lambda\gamma^{n}}
\]

\item Case $\Xi{}_{n-1}=\Lambda.$ In this case we have
\[
{\cal Q}_{P,\Xi{}_{n-1}}^{\perp}V_{i}(P)\widehat{\Psi}_{P,\Xi{}_{n-1}}=\frac{P_{i}}{\sqrt{P^{2}+m^{2}}}{\cal Q}_{P,\Xi{}_{n-1}}^{\perp}\widehat{\Psi}_{P,\Xi{}_{n-1}}=0.
\]

\end{enumerate}
For both cases the estimate
\[
\left\Vert {\cal Q}_{P,\Xi{}_{n-1}}^{\perp}\frac{1}{H_{P,\Xi{}_{n-1}}-E_{P,n-1}}{\cal Q}_{P,\Xi{}_{n-1}}^{\perp}V_{i}(P)\widehat{\Psi}_{P,\Xi{}_{n-1}}\right\Vert \leq\frac{Cg^{\epsilon}}{\Lambda\gamma^{n}}
\]
holds true. 

Choosing $\epsilon=\frac{1}{2}$ and collecting all the remainders
the bound in (\ref{eq:bad_term_est}) is seen to be true. Hence, we
have also proven the inequality in (\ref{eq:B_assumption}). This
concludes the proof of the bound in (\ref{eq:J_est}). 
\end{proof}

\end{document}